\pdfoutput=1
\PassOptionsToPackage{table}{xcolor}
\documentclass[acmsmall,nonacm]{acmart}
\usepackage{xcolor}
\usepackage{multirow}
\usepackage{siunitx}
\usepackage{makecell}
\usepackage{subcaption}
\usepackage[linesnumbered,ruled]{algorithm2e}
\usepackage[capitalise]{cleveref}
\usepackage{bbm}
\usepackage{tabularx}
\usepackage{booktabs}
\usepackage{mathtools}
\usepackage{soul}
\usepackage{latexsym}
\usepackage[most]{tcolorbox}

\usepackage{rotating}
\usepackage{autonum}

\usepackage{pgfplots}
\pgfplotsset{compat=1.15}
\usepgfplotslibrary{
  statistics,
  colorbrewer,
  groupplots,
}
\usetikzlibrary{
  patterns,
  shapes.geometric,
  decorations.text,
  matrix,
  fit,
  backgrounds,
  positioning,
}
\tikzset{
  fignode/.style={
    outer sep=0.25em,
  }
}
\tikzset{
  framedfignode/.style={
    outer sep=0.25em,
    inner sep=0.5em,
    rounded corners,
    draw,
  }
}

\newcommand{\bert}{\textsc{BERT}}

\newcommand{\ance}{\textsc{ANCE}}
\newcommand{\aggr}{\textsc{Aggretriever}}
\newcommand{\colbert}{\textsc{ColBERT}}
\newcommand{\colberter}{\textsc{ColBERTer}}
\newcommand{\tct}{\textsc{TCT-ColBERT}}
\newcommand{\cocondenser}{\textsc{coCondenser}}
\newcommand{\tildeone}{\textsc{TILDE}}
\newcommand{\tildetwo}{\textsc{TILDEv2}}
\newcommand{\splade}{\textsc{SPLADE}}
\newcommand{\dsplade}{\textsc{DistilSPLADE-max}}
\newcommand{\spade}{\textsc{SpaDE}}
\newcommand{\led}{\textsc{LED}}
\newcommand{\deepimpact}{\textsc{DeepImpact}}
\newcommand{\clear}{\textsc{CLEAR}}
\newcommand{\coil}{\textsc{COIL}}
\newcommand{\coilcr}{\textsc{COILcr}}
\newcommand{\bm}{\textsc{BM25}}
\newcommand{\deepct}{\textsc{DEEP-CT}}
\newcommand{\bertcls}{\textsc{BERT-CLS}}

\newcommand{\powerbert}{\textsc{PoWER-BERT}}
\newcommand{\skipbert}{\textsc{SkipBERT}}
\newcommand{\deebert}{\textsc{DeeBERT}}

\newcommand{\sparseretrieval}{\textsc{Sparse Retrieval}}
\newcommand{\denseretrieval}{\textsc{Dense Retrieval}}
\newcommand{\hybrid}{\textsc{Hybrid Retrieval}}
\newcommand{\reranking}{\textsc{Re-Ranking}}
\newcommand{\interpolatedreranking}{\textsc{Interpolation}}
\newcommand{\pyserini}{\textsc{Pyserini}}
\newcommand{\faiss}{\textsc{FAISS}}

\newcommand{\fastforward}{\textsc{Fast-Forward}}
\newcommand{\selbert}{\textsc{Selective BERT}}
\newcommand{\bertbase}{$\text{BERT}_\text{base}$}

\newcommand{\maxscore}{\textsc{MaxScore}}
\newcommand{\wand}{\textsc{WAND}}

\newcommand{\msmpsgdev}{\textsc{MSM-Psg-Dev}}
\newcommand{\trecdlpsgn}{\textsc{TREC-DL-Psg'19}}
\newcommand{\trecdlpsgt}{\textsc{TREC-DL-Psg'20}}

\newcommand{\trecdldocn}{\textsc{TREC-DL-Doc'19}}
\newcommand{\trecdldoct}{\textsc{TREC-DL-Doc'20}}

\newcommand{\beirmsm}{\textsc{MS MARCO}}
\newcommand{\beirfever}{\textsc{Fever}}
\newcommand{\beirfiqa}{\textsc{FiQA}}
\newcommand{\beirquora}{\textsc{Quora}}
\newcommand{\beirhpqa}{\textsc{HotpotQA}}
\newcommand{\beirdbp}{\textsc{DBpedia-Entity}}
\newcommand{\beirscifact}{\textsc{SciFact}}
\newcommand{\beirnfc}{\textsc{NFCorpus}}

\DeclareMathOperator*{\kargmax}{\textit{k}\text{-}argmax}
\newcommand{\mat}[1]{\mathbf{#1}}

\colorlet{plotColorNeutral}{gray}
\definecolor{plotColor1}{HTML}{e41a1c}
\definecolor{plotColor2}{HTML}{377eb8}
\definecolor{plotColor3}{HTML}{4daf4a}
\definecolor{plotColor4}{HTML}{984ea3}
\colorlet{plotColorNeutral*}{plotColorNeutral!60}
\colorlet{plotColor1*}{plotColor1!60}
\colorlet{plotColor2*}{plotColor2!60}
\colorlet{plotColor3*}{plotColor3!60}
\colorlet{plotColor4*}{plotColor4!60}
\pgfplotsset{
    colormap={greenred}{HTML=(4daf4a) HTML=(e41a1c)},
    colormap={redgreen}{HTML=(e41a1c) HTML=(4daf4a)}
}

\newtcbox{\gpu}[1][]{
    size=fbox,
    boxsep=1.5pt,
    on line,
    colframe=plotColor1!50,
    colback=plotColor1!50,
    enhanced,
    borderline={0.25pt}{0pt}{dashed},
    #1
}
\newtcbox{\cpu}[1][]{
    size=fbox,
    boxsep=1.5pt,
    on line,
    colframe=plotColor3!50,
    colback=plotColor3!50,
    enhanced,
    borderline={0.25pt}{0pt}{dotted},
    #1
}

\newcommand{\hide}[1]{\textcolor{gray}{#1}}
\newcommand{\tablearrow}{{} \rotatebox[origin=c]{180}{$\Lsh$} {}}
\newcommand{\midrulesep}{
    \arrayrulecolor{gray}
    \midrule[0.25pt]
    \arrayrulecolor{black}
}

\newcommand{\sigdef}[1]{{\scriptsize \texttt{[#1]}}}
\newcommand{\sigimpr}[1]{\scriptsize \textsuperscript{\texttt{[#1]}}}

\begin{document}
\title{Efficient Neural Ranking using Forward Indexes and Lightweight Encoders}

\author{Jurek Leonhardt}
\affiliation{
  \institution{Delft University of Technology}
  \city{Delft}
  \country{The Netherlands}
}
\email{L.J.Leonhardt@tudelft.nl}
\affiliation{
  \institution{L3S Research Center}
  \city{Hannover}
  \country{Germany}
}
\email{leonhardt@L3S.de}

\author{Henrik Müller}
\affiliation{
  \institution{L3S Research Center}
  \city{Hannover}
  \country{Germany}
}
\email{hmueller@L3S.de}

\author{Koustav Rudra}
\affiliation{
  \institution{Indian Institute of Technology Kharagpur}
  %\city{Kharagpur}
  \country{India}
}
\email{krudra@cai.iitkgp.ac.in}

\author{Megha Khosla}
\affiliation{
  \institution{Delft University of Technology}
  \city{Delft}
  \country{Netherlands}
}
\email{M.Khosla@tudelft.nl}

\author{Abhijit Anand}
\affiliation{
  \institution{L3S Research Center}
  \city{Hannover}
  \country{Germany}
}
\email{aanand@L3S.de}

\author{Avishek Anand}
\affiliation{
  \institution{Delft University of Technology}
  \city{Delft}
  \country{Netherlands}
}
\email{avishek.anand@tudelft.nl}

\begin{abstract}
  Dual-encoder-based dense retrieval models have become the standard in IR. They employ large Transformer-based language models, which are notoriously inefficient in terms of resources and latency.

  We propose \fastforward{} indexes---vector forward indexes which exploit the semantic matching capabilities of dual-encoder models for efficient and effective re-ranking. Our framework enables re-ranking at very high retrieval depths and combines the merits of both lexical and semantic matching via score interpolation. Furthermore, in order to mitigate the limitations of dual-encoders, we tackle two main challenges: Firstly, we improve computational efficiency by either pre-computing representations, avoiding unnecessary computations altogether, or reducing the complexity of encoders. This allows us to considerably improve ranking efficiency and latency. Secondly, we optimize the memory footprint and maintenance cost of indexes; we propose two complementary techniques to reduce the index size and show that, by dynamically dropping irrelevant document tokens, the index maintenance efficiency can be improved substantially.

  We perform evaluation to show the effectiveness and efficiency of \fastforward{} indexes---our method has low latency and achieves competitive results without the need for hardware acceleration, such as GPUs.
\end{abstract}

\begin{CCSXML}
  <ccs2012>
  <concept>
  <concept_id>10002951.10003317.10003338</concept_id>
  <concept_desc>Information systems~Retrieval models and ranking</concept_desc>
  <concept_significance>500</concept_significance>
  </concept>
  <concept>
  <concept_id>10002951.10003317.10003365.10003366</concept_id>
  <concept_desc>Information systems~Search engine indexing</concept_desc>
  <concept_significance>300</concept_significance>
  </concept>
  <concept>
  <concept_id>10002951.10003317.10003365.10003367</concept_id>
  <concept_desc>Information systems~Search index compression</concept_desc>
  <concept_significance>300</concept_significance>
  </concept>
  </ccs2012>
\end{CCSXML}

\ccsdesc[500]{Information systems~Retrieval models and ranking}
\ccsdesc[300]{Information systems~Search engine indexing}
\ccsdesc[300]{Information systems~Search index compression}

\keywords{information retrieval, IR, ranking, dual-encoders, latency, efficiency}

\maketitle
\section{Introduction}
\label{sec:intro}
Neural rankers are typically based on large pre-trained language models, the most popular example being BERT~\cite{devlin2019bert}. Due to their architectural inductive bias (like self-attention units) and complexity, these models are able to capture the semantics of documents very well, mitigating the limitations of lexical retrievers. However, their capabilities come at a price, as the models commonly used often have upwards of hundreds of millions of parameters. This makes training and even inference without specialized hardware infeasible, and it is impossible to rank all documents in a large corpus in reasonable time. Furthermore, the resources required to run these models produce a considerable amount of emissions, creating a negative impact on the environment~\cite{scells2022reduce}.

There are two predominant approaches to deal with the inefficiency of neural ranking models. The first one, referred to as \emph{retrieve-and-re-rank}~\cite{simmons1965answering,gupta2018retrieve}, uses an efficient lexical retriever to obtain a candidate set of documents for the given query. The idea is to maximize the recall, i.e., capture most of the relevant documents, in the first stage. Afterwards, the second stage employs a complex neural ranker, which \emph{re-ranks} the documents in the candidate set, in order to promote the relevant documents to higher ranks. However, the retrieve-and-re-rank approach typically employs cross-attention re-rankers, which are expensive to compute even for a small set of candidate documents. This limits the first-stage retrieval depth, as low latency is essential for many applications (e.g., search engines).

The second approach skips the lexical retrieval step entirely and uses neural models for retrieval. The \emph{dual-encoder} architecture employs a \emph{query encoder} and a \emph{document encoder}, both of which are neural models which map their string inputs to dense representations in a common vector space. Retrieval is then performed as a $k$-nearest-neighbor ($k$NN) search operation to find the documents whose representations are most similar to the query. This is referred to as \emph{dense retrieval}~\cite{karpukhin2020dense}. Representing queries and documents independently means that most 
of the computationally expensive processing happens during the indexing stage, where document representations are pre-computed. However, dense retrieval is still slower than lexical retrieval and benefits from GPU acceleration, because the query needs to be encoded during the query-processing phase. Furthermore, we find that dense retrievers generally have lower recall than term-matching-based models at higher retrieval depths.

In this paper, we argue that neither of the two approaches is ideal. Instead, our first key idea is to explore the utility of dual-encoders in the re-ranking phase instead of the retrieval phase. Using dual-encoders in the re-ranking phase allows for a drastic reduction of query processing times and resource utilization (i.e., GPUs) during document encoding. Towards this, we first show that simple interpolation-based re-ranking that combines the benefits of lexical (computed using sparse retrieval) and semantic (computed using dual-encoders) similarity can result in competitive and sometimes better performance than using cross-attention. We propose a novel index structure called \fastforward{} indexes, which exploits the ability of dual-encoders to pre-compute document representations, in order to substantially improve the runtime efficiency of re-ranking. We empirically establish that dual-encoder models show great performance as re-rankers, even though they do not use cross-attention. 

Our second observation is that most current dual-encoder models use the same encoder for both documents and queries. While this design decision makes training easier, it also means that queries have to be encoded during runtime using a, potentially expensive, forward pass. We argue that this is suboptimal; rather, queries, which are often short and concise, do not require a complex encoder to compute their representations. We propose lightweight query encoders, some of which do not contain any self-attention layers, and show that they still perform well as re-rankers, while requiring only a fraction of the resources and time. In this work, we propose two families of lightweight query encoders to drastically reduce query-encoding costs without compromising ranking performance.

Lastly, we focus on the aspects of \emph{index footprint} and \emph{index maintenance}. Since dense indexes store the pre-computed representations of documents in the corpus, they exhibit much higher storage and memory requirements compared to sparse indexes~\cite{hofstatter2022are}. At the same time, maintaining the index, i.e., adding new documents, requires expensive forward passes of the document encoder. We propose two means of reducing the memory footprint: On the one hand, we propose \emph{sequential coalescing} to compress an index by reducing the number of vectors that need to be stored; on the other hand, we experiment with choosing a smaller number of dimensions, which reduces the size of each vector. Finally, we propose efficient document encoders, which dynamically drop irrelevant tokens prior to indexing using a very simple technique. 

Our research questions are as follows:
\begin{enumerate}
    \item[\bf RQ1] How suitable are dual-encoder models for interpolation-based re-ranking in terms of performance and efficiency?
    \item[\bf RQ2] Can the re-ranking efficiency be improved by limiting the number of \fastforward{} look-ups?
    \item[\bf RQ3] To what extent does query encoder complexity affect re-ranking performance?
    \item[\bf RQ4] What is the trade-off between \fastforward{} index size and ranking performance?
    \item[\bf RQ5] Can the indexing efficiency be improved by removing irrelevant document tokens?
\end{enumerate}

We conduct extensive experimentation on existing ranking benchmarks and find that dual-encoder models are very suitable for interpolation-based re-ranking and exhibit highly desirable performance and efficiency trade-offs. We show that, with further optimizations (\emph{early stopping}---cf.\ \cref{sec:ff_indexes.early_stopping}), re-ranking efficiency can be greatly improved by limiting the number of \fastforward{} look-ups. Additionally, we report a good trade-off between \fastforward{} index size and ranking performance by using our novel \emph{sequential coalescing} algorithm (cf.\ \cref{sec:ff_indexes.coalescing}). Our experiments show that we can indeed train extremely lightweight query encoders without adversely affecting ranking performance. Specifically, our most lightweight query encoders are orders of magnitude faster than \bertbase{} models with little performance degradation. More importantly, we can migrate query-processing to CPUs instead of relying on GPUs, improving on the environmental impact. Finally, we show that we can reduce index maintenance costs by around \num{50}\% by dynamically removing irrelevant document tokens. Our code is publicly available.

Note that this paper extends our previously published work~\cite{leonhardt2022efficient}, where we introduced \fastforward{} indexes along with the \emph{sequential coalescing} and \emph{early stopping} techniques. This paper introduces the following new aspects:
\begin{enumerate}
    \item We identify the query encoder as an efficiency bottleneck of \fastforward{} indexes and propose lightweight query encoders.
    \item We show that the dimensionality of queries and documents can be reduced in order to reduce index size and compute dot products faster.
    \item We propose a \emph{selective document encoder} that dynamically identifies irrelevant document tokens and drops them prior to indexing, reducing index maintenance cost.
    \item We perform additional experiments, including analyses of the trade-offs between efficiency and performance. We discuss the limitations of our method and its out-of-domain performance.
\end{enumerate}

\section{Related Work}
\label{sec:related_work}
Classical ranking approaches, such as \bm{}~\cite{robertson2009probabilistic} or the query likelihood model~\cite{lavrenko2001relevance}, rely on the inverted index that stores term-level statistics like term frequency, inverse document frequency and positional information. We refer to this style of methods as \emph{sparse}, since it assumes sparse document representations. The recent success of large pre-trained language models (e.g., BERT) shows that \emph{semantic} or contextualized information is essential for many language tasks. In order to incorporate such information in the relevance measurement, \citet{dai2020context,dai2020context2} proposed \deepct{}, which stores contextualized scores for terms in the inverted index for text ranking. \splade{}~\cite{formal2021splade} aims to enrich sparse document representations using a trained contextual Transformer model and sparsity regularization on the term weights. Similarly, \deepimpact{}~\cite{mallia2021learning} enriches the document collection with expansion terms to learn improved term impacts. In our work, we employ efficient sparse models for high-recall first-stage retrieval and perform re-ranking using semantic models in a subsequent step.

The ability to accurately determine semantic similarity is essential in order to alleviate the vocabulary mismatch problem~\cite{mitra2019incorporating,dai2019evaluation,dai2020context2,macavaney2020expansion,mackenzie2020efficiency}. Computing the semantic similarity of a document given a query has been heavily researched in IR using smoothing methods~\cite{lafferty2001document}, topic models~\cite{wei2006lda}, embeddings~\cite{mitra2016dual}, personalized models~\cite{luxenburger2008matching}, etc. In these classical approaches, ranking is performed by interpolating the semantic similarity scores with the lexical matching scores from the first-stage retrieval. More recently, \emph{dense} neural ranking methods, which employ large pre-trained language models, have become increasingly popular. Dense rankers do not explicitly model terms, but rather compute low-dimensional dense vector representations through self-attention mechanisms in order to estimate relevance; this allows them to perform semantic matching. However, the inherent complexity of dense ranking models usually has a negative impact on latency and cost, especially with large corpora. Therefore, besides performance, efficiency has been another major concern in developing neural ranking models.

There are two common architectures of dense ranking models: \emph{Cross-attention} models take a concatenation of a query and a document as input. This allows them to perform query-document attention in order to compute the corresponding relevance score. These models are typically used as re-rankers. \textit{Dual-encoder models} employ two language models to independently encode queries and documents as fixed-size vector representations. Usually, a similarity metric between query and document vector determines their relevance. As a result, dual-encoders are mostly used for dense retrieval, but also, less commonly, for re-ranking.

We divide the remainder of the related work section into subcategories for cross-attention models, dual-encoder models, and \textit{hybrid models}, which employ both lexical and semantic rankers. Finally, we briefly cover inference efficiency for BERT-based models.

\subsection{Cross-Attention Models}
The majority of cross-attention approaches have been dominated by large contextual models~\cite{dai2019deeper,macavaney2019cedr,yilmaz2019cross,hofstatter2020interpretable,hofstatter2021efficiently,li2020parade}. The input to these ranking models is a concatenation of the query and document. This combined input results in higher query processing times, since each document has to be processed in conjugation with the query string. Thereby, cross-attention models usually re-rank a relatively small number of potentially relevant candidates retrieved in the first stage by efficient sparse methods. The expensive re-ranking computation cost is then proportional to the retrieval depth (e.g., \num{1000} documents).

Another key limitation of using cross-attention models for document ranking is the maximum acceptable number of input tokens for Transformer models, which exhibit quadratic complexity w.r.t.\ input length. Some strategies address this limitation by document truncation~\cite{macavaney2019cedr}, or chunking documents into passages~\cite{dai2019deeper,rudra2020distant}. However, the performance of chunking-based strategies depends on the chunking properties, i.e., passage length or overlap among consecutive passages~\cite{rudra2021indepth}. Recent proposals include a two-stage approach, where a query-specific summary is generated by selecting relevant parts of the document, followed by re-ranking strategies over the query and summarized document~\cite{li2021keybld,hofstatter2021intra,leonhardt2023extractive,li2023power}. Due to the efficiency concerns, we do not consider cross-attention methods in our work, but focus on dual-encoders instead.

\subsection{Dual-Encoders}
Dual-encoders learn dense vector representations for queries and documents using contextual models~\cite{karpukhin2020dense,khattab2020colbert}. The dense vectors are then indexed in an offline phase~\cite{johnson2021billion}, where retrieval is akin to performing an approximate nearest neighbor (ANN) search given a vectorized query. This allows dual-encoders to be used for both retrieval and re-ranking. Consequently, there has been a large number of follow-up works that boost the performance of dual-encoder models by improving pre-training~\cite{chang2020pre,gao2021condenser,gao2022unsupervised,lassance2023experimental,wang2022text}, optimization~\cite{gao2020complement}, and negative sampling~\cite{prakash2021learning,xiong2021approximate,zhan2021optimizing} techniques, or employing distillation approaches~\cite{lin2020distilling,zhou2022fine,liu2022adam}. \citet{lindgren2021efficient} propose a \emph{negative cache} that allows for efficient training of dual-encoder models. \led{}~\cite{zhang2022led} uses a \splade{} model to enrich a dense encoder with lexical information. \citet{lin2023aggretriever} propose \aggr{}, a dual-encoder model which aggregates and exploits all token representations (instead of only the classification token). In this work, we use dual-encoders for computing semantic similarity between queries and passages. Some approaches have also proposed architectural modifications to the aggregations between the query and passage embeddings~\cite{chen2021co,jang2021uhd,hofstatter2021efficiently}. \citet{nogueira2019from} propose a simple document expansion model. We use dual-encoder models to perform efficient semantic re-ranking in our work.

Efficiency improvements of dual-encoder-based ranking and retrieval focus mostly on either inference efficiency of the encoders or memory footprint of the indexes. \tildeone~\cite{zhuang2021tilde} and \tildetwo{}~\cite{zhuang2021fast} efficiently re-rank documents using a deep query and document likelihood model instead of a query encoder. The \spade{} model~\cite{choi2022spade} employs a \emph{dual document encoder} that has a \emph{term weighting} and \emph{term expansion} component; it improves inference efficiency by using a vastly simplified query representation. \citet{li2022citadel} employ \emph{dynamic lexical routing} in order to reduce the number of dot products in the late interaction step. \citet{cohen2022sdr} use auto-encoders to compress document representations into fewer dimensions in order to reduce the overall size. \citet{dong2022seine} propose an approach to split documents into variable-length segments and dynamically merge them based on similarity, such that each document has the same number of segments prior to indexing. \citet{hofstatter2022introducing} introduce \colberter{}, an extension of \colbert{}~\cite{khattab2020colbert}, which removes irrelevant word representations in order to reduce the number of stored vectors. In a similar fashion, \citet{lassance2022learned} propose a \emph{learned token pruning} approach, which is also used to reduce the size of \colbert{} indexes by dropping tokens that are deemed irrelevant. \citet{yang2022compact} propose a \emph{contextual quantization} approach for pre-computed document representations (such as the ones used by \colbert{}) by compressing document-specific representations of terms.

In most of the previous work, dual-encoders are used in a \emph{homogeneous} or \emph{symmetric} fashion, meaning that both the query and document encoder have the same architecture or even share weights (\emph{Siamese} encoders). \citet{jung2022semi} show that the characteristics of queries and documents are different and employ \emph{light fine-tuning} in order to adapt each encoder to its specific role. \citet{seungyeon2023embeddistill} use model distillation for asymmetric dual-encoders, where the query encoder has fewer parameters than the document encoder. \citet{lassance2022efficiency} separate the query and document encoder of \splade{} models in order to improve efficiency. In this work, we explore the use of light-weight query encoders for more efficient re-ranking.

\subsection{Hybrid Models}
Hybrid models combine sparse and dense retrieval. The most common approach is a simple linear combination of both scores~\cite{lin2020distilling}. \clear{}~\cite{gao2020complement} takes the relevance of the lexical retriever into account in the loss function of the dense retriever. \coil{}~\cite{gao2021coil} performs contextualized exact matching using pre-computed document token representations. \coilcr{}~\cite{fan2023coilcr} extends this approach by factorizing token representations and approximating them using canonical representations in order to make retrieval more efficient.

Unlike classical methods, where score interpolation is the norm, semantic similarity from neural contextual models (e.g., cross-attention or dual-encoders) is not consistently combined with the matching score. Recently, \citet{wang2021bert} showed that the interpolation of BERT-based models and lexical retrieval methods can boost the performance. Furthermore, they analyze the role of interpolation in BERT-based dense retrieval strategies and find that dense retrieval alone is not enough, but interpolation with \bm{} scores is necessary. Similarly, \citet{askari2023injecting} find that even providing the \bm{} score as part of the input text improves the re-ranking performance of BERT models.

\subsection{Inference Efficiency}
Several methods have been proposed to improve the inference efficiency of large Transformer-based models, which have quadratic time complexity w.r.t.\ the input length. \powerbert{}~\cite{goyal2020power} progressively eliminates word vectors in the subsequent encoder layers in order to reduce the input size. \deebert{}~\cite{xin2020deebert} implements an \emph{early-exit} mechanism, which may stop the computation after any Transformer layer based on the entropy of its output distribution. \skipbert{}~\cite{wang2022skipbert} uses a technique where intermediate Transformer layers can be skipped dynamically using pre-computed look-up tables. We use a simple \selbert{} approach which dynamically removes irrelevant document tokens in order to make document encoding more efficient.

\section{Preliminaries}
\label{sec:prelims}
In this section, we introduce core concepts that are essential to this work, such as retrieval, re-ranking, and interpolation.

\subsection{Interpolation-Based Re-Ranking}
\label{sec:prelims.interpolated_reranking}
The retrieval of documents or passages given a query often happens in two stages~\cite{simmons1965answering}: In the first stage, a term frequency-based (\textbf{sparse}) retrieval method (such as \bm{}~\cite{robertson1995okapi}) retrieves a set of documents from a large corpus. In the second stage, another model, which is usually much more computationally expensive, \textbf{re-ranks} the retrieved documents again.

In \textbf{sparse retrieval}, we denote the top-$k_S$ documents retrieved from the sparse index for a query $q$ as $K^q_S$. The sparse score of a query-document pair $(q, d)$ is denoted by $\phi_S(q, d)$. For the \textbf{re-ranking} part, we focus on self-attention models (such as BERT~\cite{devlin2019bert}) in this work. These models operate by creating (internal) high-dimensional dense representations of queries and documents, focusing on their semantic structure. We refer to the outputs of these models as \textbf{dense} or \textbf{semantic} scores and denote them by $\phi_D(q, d)$. Due to the quadratic time complexity of self-attention w.r.t.\ the document length (and decreasing performance with increasing document length~\cite{luan2021sparse}), long documents are often split into passages, and the score of a document is then computed as the maximum of its passage scores:
\begin{equation}
    \label{eqn:prelims.interpolated_reranking.maxp}
    \phi_D(q, d) = \max_{p_i \in d} {\phi_D(q, p_i)}
\end{equation}
This approach is referred to as \emph{maxP}~\cite{dai2019deeper}.

The retrieval approach for a query $q$ starts by retrieving $K^q_S$ from the sparse index. For each retrieved document $d \in K^q_S$, the corresponding dense score $\phi_D(q, d)$ is computed. This dense score may then be used to re-rank the retrieved set to obtain the final ranking. However, it has been shown that the scores of the sparse retriever, $\phi_S$, can be beneficial for re-ranking as well~\cite{yilmaz2019cross}. To that end, an interpolation approach is employed~\cite{bruch2023analysis}, where the final score of a query-document pair is computed as
\begin{equation}
    \label{eqn:prelims.interpolated_reranking.interpolation}
    \phi(q, d) = \alpha \cdot \phi_S(q, d) + (1 - \alpha) \cdot \phi_D(q, d).
\end{equation}
Setting $\alpha = 0$ recovers the standard re-ranking procedure.

Since the set of documents retrieved by the sparse model is typically large (e.g., $k_S = 1000$), computing the dense score for each query-document pair can be very computationally expensive. In this paper, we focus on efficient implementations of interpolation-based re-ranking, specifically the computation of the dense scores $\phi_D$.

\subsection{Dual-Encoder Models}
\label{sec:prelims.dual_encoders}
The \emph{dual-encoder} architecture~\cite{karpukhin2020dense} employs neural semantic models to compute \emph{dense vector representations} of queries and documents. Specifically, a \emph{query encoder} $\zeta$ and a \emph{document encoder} $\eta$ map queries and documents to representations in a common $a$-dimensional vector space. The relevance score $\phi_D(q, d)$ of a query-document pair is then computed as the similarity of their vector representations. A common choice for the similarity function is the dot product, such that
\begin{equation}
    \phi_D(q, d) = \zeta(q) \cdot \eta(d),
\end{equation}
where $\zeta(q), \eta(d) \in \mathbb{R}^a$.

\subsubsection{Dense Retrieval}
\label{sec:prelims.dual_encoders.dense_retrieval}
Dual-encoder models are commonly utilized to perform \textbf{dense retrieval}~\cite{karpukhin2020dense}. A \emph{dense index} contains pre-computed vector representations $\eta(d)$ for all documents $d$ in the corpus $\mathcal{D}$. To retrieve a set of documents $K^q_D$ for a query $q$, a $k$-nearest neighbor ($k$NN) search is performed to find the documents whose representations are most similar to the query:
\begin{equation}
    K^q_D = \kargmax_{1 \leq i \leq |\mathcal{D}|} (\zeta(q) \cdot \eta(d_i))
\end{equation}
In order to make dense retrieval more efficient, \emph{approximate nearest neighbor} (ANN) search is commonly employed~\cite{johnson2021billion,malkov2020efficient}. ANN search can be further accelerated using special hardware, such as GPUs~\cite{johnson2021billion}.

\subsubsection{Training}
\label{sec:prelims.dual_encoders.training}
In contrast to \emph{cross-encoder} models, which are often used for re-ranking (cf.\ \cref{sec:prelims.interpolated_reranking}), dual-encoders encode the query and document \emph{independently}, i.e., there is no query-document attention. Typically, dual-encoders for retrieval are trained using a \emph{contrastive} loss function~{\cite{karpukhin2020dense},
\begin{equation}
    \label{eq:prelims.dual_encoders.training.loss}
    \mathcal{L}(q, d^+, D^-) =
    -\log
    \left(
    \frac
        {\exp \left(\phi(q, d^+; \theta) / \tau \right)}
        {\sum_{d \in D^- \cup \{d^+\}} \exp \left(\phi(q, d; \theta) / \tau \right)}
    \right),
\end{equation}
where a training instance consists of a query $q$, a positive (relevant) document $d^+$, and a set $D^-$ of negative (irrelevant) documents. The temperature $\tau$ is a hyperparameter. Since it is usually infeasible to include all negative documents for a query in $D^-$, there are various \emph{negative sampling} approaches, such as distillation~\cite{lin2020distilling}, asynchronous indexes~\cite{xiong2021approximate}, or negative caches~\cite{lindgren2021efficient}. In this work, we use a simple \emph{in-batch} strategy~\cite{karpukhin2020dense}, where, for a query $q$, $D^-$ contains a number of \emph{hard negatives} (retrieved by BM25) along with all documents from the other queries in the same training batch.

\subsection{Hybrid Retrieval}
\label{sec:prelims.hybrid_retrieval}
\textbf{Hybrid retrieval}~\cite{gao2020complement,lin2020distilling} is similar to interpolation-based re-ranking (cf.\ \cref{sec:prelims.interpolated_reranking}). The key difference is that the dense scores $\phi_D(q, d)$ are not computed for all query-document pairs. Instead, $\phi_D$ is a dense retrieval model (cf.\ \cref{sec:prelims.dual_encoders.dense_retrieval}), which retrieves documents $d_i$ and their scores $\phi_D(q, d_i)$ using nearest neighbor search given a query $q$. A hybrid retriever combines the retrieved sets of a sparse and a dense retriever.

For a query $q$, we retrieve two sets of documents, $K^q_S$ and $K^q_D$, using the sparse and dense retriever, respectively. Note that the two retrieved sets are usually not equal.
One strategy proposed in~\cite{lin2020distilling} ranks all documents in $K^q_S \cup K^q_D$, approximating missing scores. In our experiments, however, we found that \textbf{only} considering documents from $K^q_S$ for the final ranking and discarding the rest works well. The final score is thus computed as
\begin{equation}
    \label{eqn:prelims.hybrid_retrieval.hybrid}
    \phi(q, d) = \alpha \cdot \phi_S(q, d) + (1 - \alpha) \cdot
    \begin{cases}
        \phi_D(q, d) & d \in K^q_D    \\
        \phi_S(q, d) & d \notin K^q_D
    \end{cases}.
\end{equation}
The re-ranking step in hybrid retrieval is essentially a sorting operation over the interpolated scores and takes negligible time in comparison to standard re-ranking.

\section{\fastforward{} Indexes}
\label{sec:ff_indexes}
The hybrid approach described in \cref{sec:prelims.hybrid_retrieval} has two distinct disadvantages. Firstly, in order to retrieve $K_D^q$, an (approximate) nearest neighbor search has to be performed, which is time consuming. Secondly, some of the query-document scores are expected to be missed, leading to an incomplete interpolation, where the score of one of the retrievers needs to be approximated~\cite{lin2021batch} for a number of query-document pairs.

\begin{figure}
    \centering
    \begin{tikzpicture}
    \draw[
        thick,
        rotate=135,
        postaction={
                decorate,
                decoration={
                        text effects along path,
                        text={\ Embedding space\ },
                        text align=center,
                        text effects/.cd,
                        text along path,
                        every character/.style={
                                fill=white,
                                yshift=-0.5ex,
                            },
                    },
            },
    ] (0,0) ellipse (2.5 and 3.5);

    \node[
        circle,
        dotted,
        thick,
        minimum size=50,
        draw=red,
    ] (P012) at (1.25,1.25) {};
    \node[
        circle,
        fill=red,
        inner sep=0,
        minimum size=5,
        draw,
    ] (P) at (P012) {};
    \node[
        fill=white,
    ] at (P012.north) {\small \textcolor{red}{$\hat{p}_{012}$}};
    \draw[fill=gray] (0.75,0.75) circle (0.05) node[right] {\small $p_0$};
    \draw[fill=gray] (1.75,0.75) circle (0.05) node[above] {\small $p_1$};
    \draw[fill=gray] (1.75,1.75) circle (0.05) node[left] {\small $p_2$};

    \draw[fill=gray] (0.45,-1.7) circle (0.05) node[below] {\small $p_3$};
    \draw[fill=gray] (-2,0) circle (0.05) node[below] {\small $p_4$};

    \node[
        circle,
        dotted,
        thick,
        minimum size=45,
        draw=red,
    ] (P56) at (-0.3,-1) {};
    \node[
        circle,
        fill=red,
        inner sep=0,
        minimum size=5,
        draw,
    ] at (P56) {};
    \node[
        fill=white,
    ] at (P56.north) {\small \textcolor{red}{$\hat{p}_{56}$}};
    \draw[fill=gray] (0.15,-1.35) circle (0.05) node[above] {\small $p_5$};
    \draw[fill=gray] (-0.75,-0.65) circle (0.05) node[below] {\small $p_6$};

    \node[draw] (FF) at (-5,2) {$\eta^{\text{FF}}(p_1)$};
    \draw[
        ->,
        dashed,
        shorten >= 2,
    ] (FF) -- (P) node[
        pos=0.4,
        sloped,
        fill=white,
    ] {Index look-up};
\end{tikzpicture}
    \caption{Sequential coalescing combines the representations of similar consecutive passages as their average. Note that $p_3$ and $p_5$ are not combined, as they are not consecutive passages.}
    \label{fig:ff_indexes.coalescing}
\end{figure}
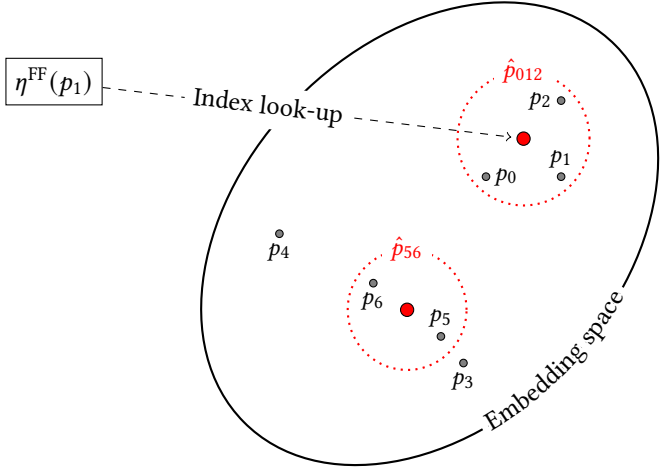
\begin{figure}
    \centering
    \begin{tikzpicture}
    \newcommand{\repr}[2]{$\begin{bmatrix} #1 \\ \vdots \\ #2 \end{bmatrix}$}
    \matrix [
        matrix of nodes,
        nodes in empty cells,
        column sep=5,
    ]
    (T) at (0,0)
    {
         &        & \node[align=center] {Matching \\ score}; &   & \node[align=center] {Semantic \\ similarity}; & \\
         & $d_8$  & \num{0.34}                               & + & \num{0.48}                                    & \\
         & $d_2$  & \num{0.32}                               & + & \textcolor{red}{\num{0.61}}        & \\
         &        &                                          &   &                                               & \\
         & $d_1$  & \num{0.08}                               & + & \textcolor{gray}{\num{0.61}}       & \\
         & \vdots & \vdots                                   &   &                                               & \\
    };

    \draw[thick] (T-1-1.south east) -- (T-1-6.south west);
    \draw[thin] (T-4-1.east) -- (T-4-6.west) node[midway,fill=white] {\scriptsize \textit{early stopping for top-\num{1}}};

    \node[
        below=1 of T-5-5,
        dashed,
        rounded corners,
        draw,
        outer sep=2,
    ] (M) {\textcolor{red}{estimated maximum}};
    \draw[->,dashed] (M) -- (T-5-5);

    \node[right=0 of T-5-5] {$\leq$ current top-$k$};

    \matrix [
        draw,
        rounded corners,
        right=3.5 of T,
        matrix of nodes,
        nodes in empty cells,
        inner sep=7,
        outer sep=2,
        column 3/.style={anchor=base west},
    ]
    (F)
    {
        $d_1$ & $\mapsto$ & \tiny \repr{0.09}{0.91}                                     \\
        $d_2$ & $\mapsto$ & \tiny \repr{0.58}{0.37} \repr{0.44}{0.19} \repr{0.71}{0.60} \\
              & \vdots    &                                                             \\
        $d_n$ & $\mapsto$ & \tiny \repr{0.12}{0.89} \repr{0.33}{0.10}                   \\
    };
    \node[
        fill=white,
    ] at (F.north) {\textbf{\fastforward{} index}};

    \draw[->,dashed] (T-3-5) -- (T-3-5 -| F.west) node[
        midway,
        fill=white,
        above,
    ] {$\max_{p_i \in d_2} \left( \zeta(q) \cdot \eta(p_i) \right)$};

\end{tikzpicture}
    \caption{Early stopping reduces the number of interpolation steps by computing an approximate upper bound for the dense scores. This example depicts the most extreme case, where only the top-\num{1} document is required.}
    \label{fig:ff_indexes.early_stopping}
\end{figure}
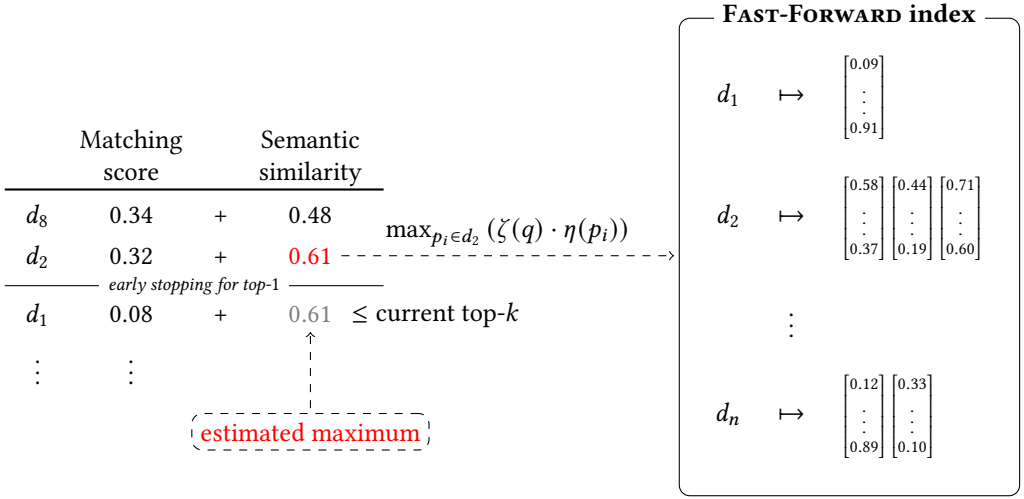
In this section, we propose \fastforward{} indexes as an efficient way of computing dense scores for known documents that alleviates the aforementioned issues. Specifically, \fastforward{} indexes build upon dual-encoder dense retrieval models that compute the score of a query-document pair as a dot product
\begin{equation}
    \phi_D(q, d) = \zeta(q) \cdot \eta(d),
\end{equation}
where $\zeta$ and $\eta$ are the query and document encoders, respectively. Examples of such models are \ance{}~\cite{xiong2021approximate} and \tct{}~\cite{lin2021batch}. Since the query and document representations are independent for two-tower models, we can pre-compute the document representations $\eta(d)$ for each document $d$ in the corpus. These document representations are then stored in an efficient hash map, allowing for look-ups in constant time. After the index is created, the score of a query-document pair can be computed as
\begin{equation}
    \phi_D^{FF}(q, d) = \zeta(q) \cdot \eta^{FF}(d),
\end{equation}
where the superscript $FF$ indicates the look-up of a pre-computed document representation in the \fastforward{} index. At retrieval time, only $\zeta(q)$ needs to be computed once for each query. As queries are usually short, this can be done on CPUs. The main benefit of this method is that the number of documents to be re-ranked can be much higher than with cross-attention models; the scoring operation is a simple look-up and dot product computation.

Note that the use of large Transformer-based query encoders still remains a bottleneck in terms of latency (or, if it is run on GPUs, cost). In \cref{sec:efficient_encoders}, we focus on lightweight encoder models.

\subsection{Index Compression via Sequential Coalescing}
\label{sec:ff_indexes.coalescing}
A major disadvantage of dense indexes and dense retrieval in general is the size of the final index. This is caused by two factors: Firstly, in contrast to sparse indexes, the dense representations cannot be stored as efficiently as sparse vectors. Secondly, the dense encoders are typically Transformer-based, imposing a (soft) limit on their input lengths due to their quadratic time complexity with respect to the inputs. Thus, long documents are split into passages prior to indexing (\emph{maxP} indexes).

\begin{algorithm}[t]
    \DontPrintSemicolon
    \SetKwFunction{Dist}{cosine\_distance}
    \SetKwFunction{Mean}{mean}
    \SetKw{In}{in}
    \KwIn{list of passage vectors $P$ (original order) of a document, distance threshold $\delta$}
    \KwOut{coalesced passage vectors $P'$}
    $P' \leftarrow$ empty list\;
    $\mathcal{A} \leftarrow \emptyset$\;
    \ForEach{$v$ \In $P$}{
        \uIf{first iteration}{
            \tcp{do nothing}
        }
        \uElseIf{$\Dist(v, \overline{\mathcal{A}}) \geq \delta$}{
            append $\overline{\mathcal{A}}$ to $P'$\;
            $\mathcal{A} \leftarrow \emptyset$\;
        }
        add $v$ to $\mathcal{A}$\;
        $\overline{\mathcal{A}} \leftarrow \Mean(\mathcal{A})$\;
    }
    append $\overline{\mathcal{A}}$ to $P'$\;
    \Return{$P'$}\;
    \caption{Compression of dense maxP indexes by sequential coalescing}
    \label{alg:ff_indexes.coalescing}
\end{algorithm}
As an increase in the index size has a negative effect on efficiency, both for nearest neighbor search and \fastforward{} indexing as used by our approach, we exploit a \emph{sequential coalescing} approach as a way of dynamically combining the representations of consecutive passages within a single document in maxP indexes. The idea is to reduce the number of passage representations in the index for a single document. This is achieved by exploiting the \emph{topical locality} that is inherent to documents~\cite{leonhardt2020boilerplate}. For example, a single document might contain information regarding multiple topics; due to the way human readers naturally ingest information, we expect documents to be authored such that a single topic appears mostly in consecutive passages, rather than spread throughout the whole document. Our approach aims to combine consecutive passage representations that encode similar information. To that end, we employ the cosine distance function and a \emph{threshold} parameter $\delta$ that controls the degree of coalescing. Within a single document, we iterate over its passage vectors in their original order and maintain a set $\mathcal{A}$, which contains the representations of the already processed passages, and continuously compute $\overline{\mathcal{A}}$ as the average of all vectors in $\mathcal{A}$. For each new passage vector $v$, we compute its cosine distance to $\overline{\mathcal{A}}$. If it exceeds the distance threshold $\delta$, the current passages in $\mathcal{A}$ are combined as their average representation $\overline{\mathcal{A}}$. Afterwards, the combined passages are removed from $\mathcal{A}$ and $\overline{\mathcal{A}}$ is recomputed. This approach is illustrated in \cref{alg:ff_indexes.coalescing}. \cref{fig:ff_indexes.coalescing} shows an example index after coalescing. To the best of our knowledge, there are no other forward index compression techniques proposed in literature so far.

\subsection{Faster Interpolation by Early Stopping}
\label{sec:ff_indexes.early_stopping}
As described in \cref{sec:prelims.interpolated_reranking}, by interpolating the scores of sparse and dense retrieval models, we perform implicit re-ranking, where the dense representations are pre-computed and can be looked up in a \fastforward{} index at retrieval time. Furthermore, increasing the sparse retrieval depth $k_S$, such that $k_S > k$, where $k$ is the final number of documents, improves the performance. A drawback of this is that an increase in the number of retrieved documents also results in an increase in the number of index look-ups.

Common term pruning mechanisms for term-at-a-time retrieval, such as \maxscore{}~\cite{turtle1995query} or \wand{}~\cite{broder2003efficient}, accelerate query processing for inverted-index-based retrievers; however, these techniques are not compatible with neural ranking models based on contextual query and document representations. Our use case is more similar to \emph{top-k query evaluation}, with algorithms such as the \emph{threshold algorithm}~\cite{fagin2001optimal} or probabilistic approximations~\cite{theobald2004topk}, but these approaches usually require sorted access, which is not available for the dense re-ranking scores in our case.

\begin{algorithm}[t]
    \DontPrintSemicolon
    \SetKwFunction{Sparse}{sparse}
    \SetKwFunction{Max}{max}
    \SetKw{In}{in}
    \SetKw{Break}{break}
    \KwIn{query $q$, sparse retrieval depth $k_S$, cut-off depth $k$, interpolation parameter $\alpha$}
    \KwOut{approximated top-$k$ scores $Q$}
    $Q \leftarrow$ priority queue of size $k$\;
    $s_{D} \leftarrow -\infty$\;
    $s_{min} \leftarrow -\infty$\;
    \ForEach{$d$ \In $\Sparse(q, k_S)$}{
        \uIf{$Q$ is full}{
            $s_{min} \leftarrow$ remove smallest item from Q\;
            $s_{best} \leftarrow \alpha \cdot \phi_S(q, d) + (1 - \alpha) \cdot s_{D}$\; \label{alg:ff_indexes.early_stopping:sbest}

            \uIf{$s_{best} \leq s_{min}$}{
                \tcp{early stopping}
                put $s_{min}$ into $Q$\;
                \Break\;
            }
        }
        \tcp{approximate max. dense score}
        $s_{D} \leftarrow \Max(\phi_D(q, d), s_{D})$\;
        $s \leftarrow \alpha \cdot \phi_S(q, d) + (1 - \alpha) \cdot \phi_D(q, d)$\;
        put $\Max(s, s_{min})$ into $Q$\;
    }
    \Return{$Q$}\;
    \caption{Interpolation with early stopping}
    \label{alg:ff_indexes.early_stopping}
\end{algorithm}
In this section, we propose an extension to \fastforward{} indexes that allows for \emph{early stopping}, i.e., avoiding a number of unnecessary look-ups, for cases where $k_S > k$ by approximating the maximum possible dense score. The early stopping approach takes advantage of the fact that documents are ordered by their sparse scores $\phi_S(q, d)$. Since the number of retrieved documents, $k_S$, is finite, there exists an upper limit $s_D$ for the corresponding dense scores such that $\phi_D(q, d) \leq s_D \forall d \in K^q_S$. Since the retrieved documents $K^q_S$ are ordered by their sparse scores, we can simultaneously perform interpolation and re-ranking by iterating over the ordered list of documents: Let $d_i$ be the $i$th highest ranked document by the sparse retriever. Recall that we compute the final score as
\begin{equation}
    \phi(q, d_i) = \alpha \cdot \phi_S(q, d_i) + (1 - \alpha) \cdot \phi_D(q, d_i).
\end{equation}
If $i > k$, we can compute the upper bound for $\phi(q, d_i)$ by exploiting the aforementioned ordering:
\begin{equation}
    s_{best} = \alpha \cdot \phi_S(q, d_{i-1}) + (1 - \alpha) \cdot s_D.
\end{equation}
In turn, this allows us to stop the interpolation and re-ranking if $s_{best} \leq s_{min}$, where $s_{min}$ denotes the score of the $k$th document in the current ranking (i.e., the currently lowest ranked document). Intuitively, this means that we stop the computation once the \emph{highest possible} interpolated score $\phi(q, d_i)$ is too low to make a difference. The approach is illustrated in \cref{alg:ff_indexes.early_stopping} and \cref{fig:ff_indexes.early_stopping}. Since the dense scores $\phi_D$ are usually unnormalized, the upper limit $s_D$ is unknown in practice. We thus approximate it by using the highest observed dense score at any given step.

\subsubsection{Theoretical Analysis}
We first show that the early stopping criteria, when using the true maximum of the dense scores, is sufficient to obtain the top-$k$ scores.
\begin{theorem}
    Let $s_D$, as used in \cref{alg:ff_indexes.early_stopping}, be the true maximum of the dense scores. Then the returned scores are the actual top-$k$ scores.
\end{theorem}
\begin{proof}
    First, note that the sparse scores, $\phi_S(q, d_i)$, are already sorted in decreasing order for a given query. By construction, the priority queue $Q$ always contains the highest scores corresponding to the list parsed so far. Let, after parsing $k$ scores, $Q$ be full. Now the possible best score $s_{best}$ is computed using the sparse score found next in the decreasing sequence and the maximum of all dense scores, $s_D$ (cf.\ \cref{alg:ff_indexes.early_stopping:sbest}). If $s_{best}$ is less than the minimum of the scores in $Q$, then $Q$ already contains the top-$k$ scores. To see this, note that the first component of $s_{best}$ is the largest among all unseen sparse scores (as the list is sorted) and $s_D$ is the maximum of the dense scores by our assumption.
\end{proof}
Next, we show that a good approximation of the top-$k$ scores can be achieved by using the sample maximum. To prove our claim, we use the Dvoretzky–Kiefer–Wolfowitz (DKW)~\cite{massart1990tight} inequality.
\begin{lemma}
    \label{lem:ff_indexes.early_stopping.dkw}
    Let $X_1, X_2, ..., X_n$ be $n$ real-valued independent and identically distributed random variables with the cumulative distribution function $F(\cdot)$. Let $F_n(\cdot)$ denote the empirical cumulative distributive function, i.e.,
    \begin{equation}
        F_{n}(x)=\frac{1}{n} \sum_{i=1}^{n} \mathbbm{1}_{\left\{X_{i} \leq x\right\}}, \quad x \in \mathbb{R}.
    \end{equation}
    According to the DKW inequality, the following estimate holds:
    \begin{equation}
        \Pr \left( \sup_{x \in \mathbb{R}} \left( F_{n}(x) - F(x) \right) > \epsilon \right) \leq e^{-2n \epsilon^2} \forall \epsilon \geq \sqrt{\frac{1}{2n} \ln 2}.
    \end{equation}
\end{lemma}
In the following, we show that, if $s_D$ is chosen as the maximum of a large random sample drawn from the set of dense scores, then the probability that any given dense score, chosen independently and uniformly at random from the dense scores, is greater than $s_D$ is exponentially small in the sample size.
\begin{theorem}
    Let $x_1, x_2, ..., x_n$ be a real-valued independent and identically distributed random sample drawn from the distribution of the dense scores with the cumulative distribution function $F(\cdot)$. Let $z = \max{(x_1, x_2, ..., x_n)}$. Then, for every $\epsilon > \frac{1}{\sqrt{2n}} \ln 2$, we obtain
    \begin{equation}
        \label{eq:ff_indexes.early_stopping.dkw_dense}
        \Pr(F(z) < 1 - \epsilon) \leq e^{-2n \epsilon^2}.
    \end{equation}
\end{theorem}
\begin{proof}
    Let $F_n(\cdot)$ denote the empirical cumulative distribution function as above. Specifically, $F_n(x)$ is equal to the fraction of variables less than or equal to $x$. We then have $F_n(z) = 1$. By \cref{lem:ff_indexes.early_stopping.dkw}, we infer
    \begin{equation}
        \Pr(F_n(z) - F(z) > \epsilon) \leq e^{-2n \epsilon^2}.
    \end{equation}
    Substituting $F_n(z) = 1$, we obtain Equation~\eqref{eq:ff_indexes.early_stopping.dkw_dense}.
\end{proof}
This implies that the probability of any random variable $X$, chosen randomly from the set of dense scores, being less than or equal to $s_D$, is greater than or equal to $1 - \epsilon$ with high probability, i.e.,
\begin{equation}
    \Pr(P_D(X \leq s_D) \geq 1 - \epsilon) \geq 1 - e^{-2n \epsilon^2},
\end{equation}
where $P_D$ denotes the probability distribution of the dense scores. This means that, as our sample size grows until it reaches $k$, the approximation improves.
Note that, in our case, the dense scores are sorted (by corresponding sparse score) and thus the i.i.d.\ assumption cannot be ensured. However, we observed that the dense scores are positively correlated with the sparse scores. We argue that, due to this correlation, we can approximate the maximum score well.

\section{Efficient Encoders}
\label{sec:efficient_encoders}
BERT models are the de facto standard for both query and document encoders~\cite{karpukhin2020dense,lin2020distilling,xiong2021approximate}. The encoders are often \emph{homogeneous}, meaning that the architectures of both models are identical, or even \emph{Siamese}, i.e., the same encoder weights are used for both queries and documents. Other approaches are \emph{semi-Siamese} models~\cite{jung2022semi}, where \emph{light fine-tuning} is used to adapt each encoder to its input characteristics, or \tildeone{}~\cite{zhuang2021tilde} and \tildetwo{}~\cite{zhuang2021fast}, which do not require dense query representations. However, the most common choice remains the use of \bertbase{} for both encoders.

\begin{figure}
    \centering
    \input{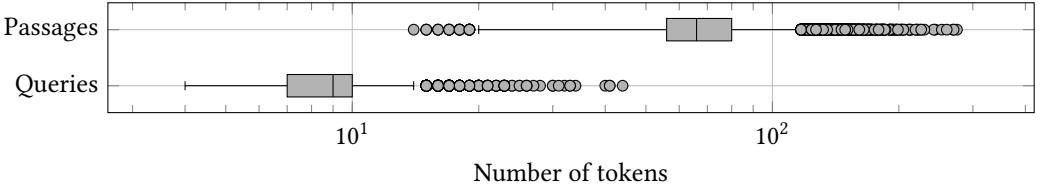}
    \caption{The distribution of query and passage lengths in the MS MARCO corpus. The statistics are computed based on the development set queries and the first \num{10000} passages from the corpus using a \bertbase{} tokenizer.}
    \label{fig:efficient_encoders.length_stats}
\end{figure}
In this paper, we argue that the homogeneous structure is not ideal for dual-encoder IR models w.r.t.\ query processing efficiency, since the characteristics of queries and documents differ~\cite{jung2022semi}. We illustrate those characteristics w.r.t.\ the average number of tokens in \cref{fig:efficient_encoders.length_stats}. This section focuses on model architectures for both query and document encoding that aim to improve the overall efficiency of the ranking process.

\subsection{Lightweight Query Encoders}
\label{sec:efficient_encoders.query}
Query encoders need to be run online during query processing, i.e., the representations cannot be pre-computed. Consequently, query encoding latency is essential for many downstream applications, such as search engines. Our experiments reveal that even encoding a large batch of \num{256} queries using a \bertbase{} model on CPU takes more than \num{3} seconds (cf.\ \cref{fig:results.query_enc.psg_results.latency_dev}), resulting in roughly \num{12} milliseconds per query (smaller batch sizes or even single queries lead to even slower encoding). Since queries are typically short and concise, we argue that query encoders require lower complexity (e.g., in terms of the number of parameters) than document encoders. Our proposed query encoders are considerably more lightweight than standard \bertbase{} models, and thus more efficient in terms of latency and resources.
\begin{figure*}
    \begin{subfigure}{.55\linewidth}
        \centering
        \begin{tikzpicture}
    \node[
        framedfignode,
        minimum width=150,
    ] (E) at (0,0) {Embedding};

    \node[
        fignode,
        below=1.25 of E.west,
        anchor=west,
    ] (T1) {$t_1$};
    \node[
        below=-0.2 of T1
    ] {\tiny \texttt{[CLS]}};

    \node[
        fignode,
        right=0 of T1,
        anchor=west,
    ] (T2) {$t_2$};
    \node[
        below=-0.2 of T2
    ] {\tiny \texttt{what}};

    \node[
        fignode,
        right=0 of T2,
        anchor=west,
    ] (T3) {$t_3$};
    \node[
        below=-0.2 of T3
    ] {\tiny \texttt{is}};

    \node[
        fignode,
        right=0 of T3,
        anchor=west,
    ] (T4) {$t_4$};
    \node[
        below=-0.2 of T4
    ] (T4b) {\tiny \texttt{the}};

    \node[
        fignode,
        anchor=north east,
    ] (Tq) at (T4.north -| E.east) {$t_{|q|}$};
    \node (Tqb) at (Tq |- T4b) {\tiny \texttt{[SEP]}};

    \node[
        fignode,
    ] at ($(T4.north)!0.5!(Tqb.south)$) {\dots};

    \node[
        framedfignode,
        minimum width=150,
        above=1.25 of E.west,
        anchor=west,
    ] (L) {Transformer encoder layer};

    \node[
        fignode,
        above=1.25 of L -| T1,
        anchor=center,
    ] (O) {$\hat{\zeta}(q)$};

    \draw[->] (T1) -- (T1 |- E.south);
    \draw[->] (T2) -- (T2 |- E.south);
    \draw[->] (T3) -- (T3 |- E.south);
    \draw[->] (T4) -- (T4 |- E.south);
    \draw[->] (Tq) -- (Tq |- E.south);
    \draw[->] (E.north -| T1) -- (L.south -| T1);
    \draw[->] (E.north -| T2) -- (L.south -| T2);
    \draw[->] (E.north -| T3) -- (L.south -| T3);
    \draw[->] (E.north -| T4) -- (L.south -| T4);
    \draw[->] (E.north -| Tq) -- (L.south -| Tq);
    \draw[->] (L.north -| T1) -- (O);

    \draw [
        decorate,
        decoration={
                brace,
                raise=5,
                amplitude=5,
            },
    ] (L.south west) -- (L.north west) node[left=0.5,midway] {$L \times$};
\end{tikzpicture}
        \caption{Attention-based query encoder}
        \label{fig:efficient_encoders.query.encoders.attn}
    \end{subfigure}
    \begin{subfigure}{.44\linewidth}
        \centering
        \begin{tikzpicture}
    \node[
        framedfignode,
        minimum width=150,
    ] (E) at (0,0) {Embedding};

    \node[
        fignode,
        below=1.25 of E.west,
        anchor=west,
    ] (T1) {$t_1$};
    \node[
        below=-0.2 of T1
    ] {\tiny \texttt{[CLS]}};

    \node[
        fignode,
        right=0 of T1,
        anchor=west,
    ] (T2) {$t_2$};
    \node[
        below=-0.2 of T2
    ] {\tiny \texttt{what}};

    \node[
        fignode,
        right=0 of T2,
        anchor=west,
    ] (T3) {$t_3$};
    \node[
        below=-0.2 of T3
    ] {\tiny \texttt{is}};

    \node[
        fignode,
        right=0 of T3,
        anchor=west,
    ] (T4) {$t_4$};
    \node[
        below=-0.2 of T4
    ] (T4b) {\tiny \texttt{the}};

    \node[
        fignode,
        anchor=north east,
    ] (Tq) at (T4.north -| E.east) {$t_{|q|}$};
    \node (Tqb) at (Tq |- T4b) {\tiny \texttt{[SEP]}};

    \node[
        fignode,
    ] at ($(T4.north)!0.5!(Tqb.south)$) {\dots};

    \node[
        framedfignode,
        trapezium,
        trapezium angle=60,
        trapezium stretches=true,
        minimum width=150,
        above=1.25 of E.center,
        anchor=center,
    ] (L) {Average};

    \node[
        fignode,
        above=1.25 of L.center,
        anchor=center,
    ] (O) {$\hat{\zeta}(q)$};

    \draw[->] (T1) -- (T1 |- E.south);
    \draw[->] (T2) -- (T2 |- E.south);
    \draw[->] (T3) -- (T3 |- E.south);
    \draw[->] (T4) -- (T4 |- E.south);
    \draw[->] (Tq) -- (Tq |- E.south);

    \draw[->] (E.north -| T1) -- (L.south -| T1);
    \draw[->] (E.north -| T2) -- (L.south -| T2);
    \draw[->] (E.north -| T3) -- (L.south -| T3);
    \draw[->] (E.north -| T4) -- (L.south -| T4);
    \draw[->] (E.north -| Tq) -- (L.south -| Tq);

    \draw[->] (L) -- (O);
\end{tikzpicture}
        \caption{Embedding-based query encoder}
        \label{fig:efficient_encoders.query.encoders.embedding}
    \end{subfigure}
    \caption{The query-encoder types used in this work. Note that the positional encoding that is added to BERT input tokens has been omitted in this figure.}
    \label{fig:efficient_encoders.query.encoders}
\end{figure*}
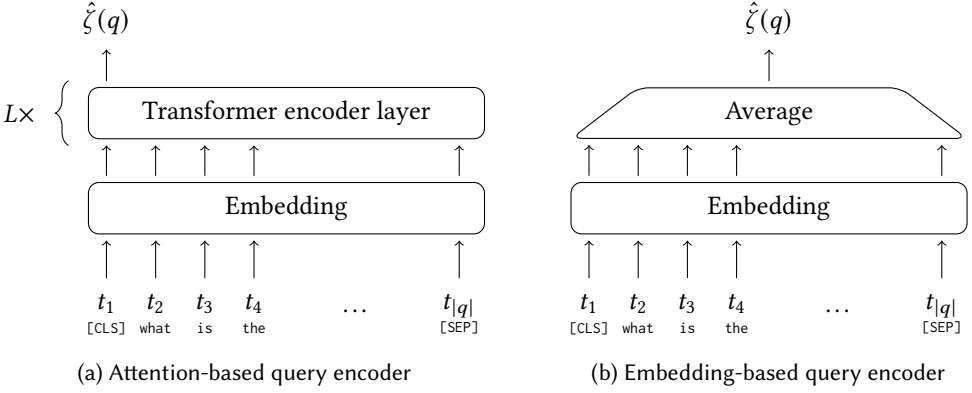

\subsubsection{Attention-based}
\label{sec:efficient_encoders.query.attn}
Attention-based query encoders (such as models based on BERT~\cite{devlin2019bert}) use Transformer encoder layers~\cite{vaswani2017attention} to compute query representations. Each of these layers has two main components --- \emph{multi-head attention} and a feed-forward sub-layer --- both of which include residual connections and layer normalization operations.

Attention is computed based on three input matrices --- the \emph{queries} $\mat{Q}$, \emph{keys} $\mat{K}$, and \emph{values} $\mat{V}$:
\begin{equation}
    \operatorname{Attn}(\mat{Q}, \mat{K}, \mat{V}) = \operatorname{softmax}\left(\frac{\mat{Q} \mat{K}^T}{\sqrt{d_k}}\right) \mat{V}.
\end{equation}
Multi-head attention computes attention multiple times (using $A$ \emph{attention heads} $h_i$) and concatenates the results, as denoted by $\circ$, i.e.,
\begin{equation}
    \begin{aligned}
        \operatorname{MultiHead}(\mat{Q}, \mat{K}, \mat{V}) & = \left(h_1 \circ \ldots \circ h_A \right) \mat{W}^{O},                                            \\
        \text{where } h_i                                   & = \operatorname{Attn} \left(\mat{Q} \mat{W}_i^Q, \mat{K} \mat{W}_i^K, \mat{V} \mat{W}_i^V \right).
    \end{aligned}
\end{equation}
The matrices $\mat{W}_i^Q \in \mathbb{R}^{H \times d_k}$, $\mat{W}_i^K \in \mathbb{R}^{H \times d_k}$, $\mat{W}_i^V \in \mathbb{R}^{H \times d_v}$ and $\mat{W}^O \in \mathbb{R}^{A d_v \times H}$ are trainable parameters, $H$ denotes the dimension of hidden representations in the model, and $d_k = \frac{H}{A}$ is a scaling factor.

Since Transformer encoders compute self-attention, the three inputs $\mat{Q}$, $\mat{K}$ and $\mat{V}$ originate from the same place, i.e., they are projections of the output of the previous encoder layer. The inputs to the first encoder layer originate from a token embedding layer. We denote the embedding operation as $E: \mathbb{N} \mapsto \mathbb{R}^H$, such that $E(t)$ is the embedding vector of a token $t$.

Given a BERT-based encoder and a query $q = (t_1, ..., t_{|q|})$, where $t_i$ are \texttt{WordPiece} tokens, the query representation is computed as
\begin{equation}
    \hat{\zeta}_\text{Attn}(q) = \operatorname{BERT}_\text{CLS}(\texttt{[CLS]}, t_1, ..., t_{|q|}, \texttt{[SEP]}),
\end{equation}
where $\operatorname{BERT}_\text{CLS}$ indicates that the output vector corresponding to the \emph{classification token}, denoted by \texttt{[CLS]}, is used. \Cref{fig:efficient_encoders.query.encoders.attn} shows attention-based query encoders.

The usual choice for query encoders, $\text{BERT}_\text{base}$, has $L=12$ layers, $H=768$ dimensions for hidden representations and $A=12$ attention heads. In this work, we investigate how less complex query encoders impact the re-ranking performance. Specifically, we vary three hyperparameters, namely the number of Transformer layers $L$, hidden dimensions $H$ and attention heads $A$. The pre-trained BERT models we use are provided by \citet{turc2019well}.

\subsubsection{Embedding-based}
\label{sec:efficient_encoders.query.embedding}
Embedding-based query encoders can be seen as a special case of BERT-based query encoders (cf.\ \cref{sec:efficient_encoders.query.attn}). Setting $L=0$, we obtain a model without any Transformer encoder layers; what's left is only the token embedding layer $E$.

Due to the omission of self-attention (and thus, contextualization) altogether, the usage of the \texttt{[CLS]} token is not feasible for this approach. Instead, a query $q = (t_1, ..., t_{|q|})$ is represented simply as the average of its token embeddings, i.e.,
\begin{equation}
    \hat{\zeta}_\text{Emb}(q) = \frac{\sum_{t_i \in q} E(t_i)}{|q|}.
\end{equation}
Embedding-based query encoders are illustrated in \cref{fig:efficient_encoders.query.encoders.embedding}.

\subsection{Selective Document Encoders}
\label{sec:efficient_encoders.doc}
Document encoders are not run during query processing time, since document representations are pre-computed and indexed. However, the computation of document representations still requires a substantial amount of time and resources. This is particularly important for applications like web search, where \emph{index maintenance} plays an important role, usually due to large amounts of new documents constantly needing to be added to the index. The effect is further amplified by the maxP approach (cf.\ \cref{eqn:prelims.interpolated_reranking.maxp}), where long documents require more than one encoding step. Since documents tend to be much longer and more complex than queries, lightweight document encoders would likely negatively affect performance, and recent research suggests that larger document encoders lead to better results~\cite{ni2021large}. However, due to the nature of documents obtained from web pages, we expect a considerable number of document tokens to be irrelevant for the encoding step; examples for this are stop words or redundant (repeated) information. Similar observations have been made in other approaches~\cite{hofstatter2022introducing}. Furthermore, recent research~\cite{rau2022role} has shown that certain aspects, such as the position of tokens, are not essential for large language models to perform well. Our proposed document encoders assign a \emph{relevance score} to each input token and dynamically drop low-scoring tokens before computing self-attention in order to make the document encoding step more efficient.

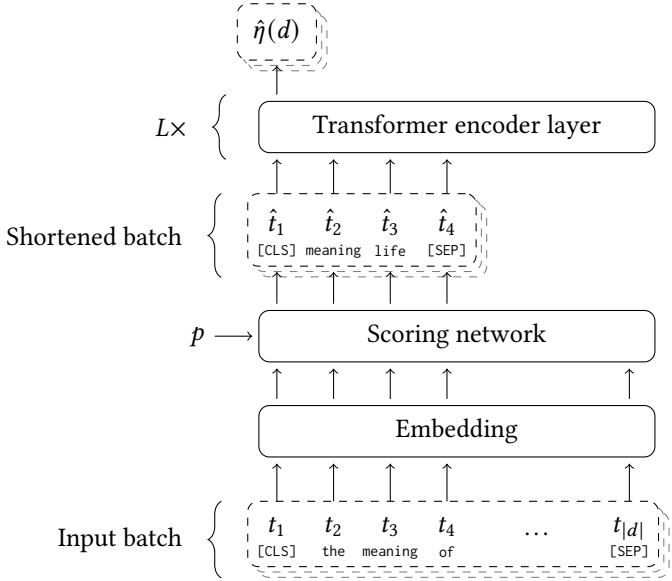
\begin{figure}
    \centering
    \begin{tikzpicture}
    % embedding
    \node[
        framedfignode,
        minimum width=150,
    ] (E) at (0,0) {Embedding};

    % inputs
    \node[
        fignode,
        below=1.25 of E.west,
        anchor=west,
    ] (T1) {$t_1$};
    \node[
        below=-0.2 of T1,
    ] (T1b) {\tiny \texttt{[CLS]}};

    \node[
        fignode,
        right=0.1 of T1,
        anchor=west,
    ] (T2) {$t_2$};
    \node[
        below=-0.2 of T2,
    ] (T2b) {\tiny \texttt{the}};

    \node[
        fignode,
        right=0.1 of T2,
        anchor=west,
    ] (T3) {$t_3$};
    \node[
        below=-0.2 of T3,
    ] {\tiny \texttt{meaning}};

    \node[
        fignode,
        right=0.1 of T3,
        anchor=west,
    ] (T4) {$t_4$};
    \node[
        below=-0.2 of T4,
    ] (T4b) {\tiny \texttt{of}};

    \node[
        fignode,
        anchor=north east,
    ] (Td) at (T4.north -| E.east) {$t_{|d|}$};
    \node (Tdb) at (Td |- T4b) {\tiny \texttt{[SEP]}};

    \node[
        fignode,
    ] at ($(T4.north)!0.5!(Tdb.south)$) {\dots};

    % scoring network
    \node[
        framedfignode,
        minimum width=150,
        above=1.25 of E.west,
        anchor=west,
    ] (S) {Scoring network};
    \node[
        left=0.5 of S,
    ] (P) {$p$};
    \draw[->] (P) -- (S);

    % intermediate representations
    \node[
        fignode,
        above=1.5 of S.west,
        anchor=west,
    ] (Ts1) {$\hat{t}_1$};
    \node[
        below=-0.2 of Ts1,
    ] (Ts1b) {\tiny \texttt{[CLS]}};

    \node[
        fignode,
        right=0.1 of Ts1,
        anchor=west,
    ] (Ts2) {$\hat{t}_2$};
    \node[
        below=-0.2 of Ts2,
    ] (Ts2b) {\tiny \texttt{meaning}};

    \node[
        fignode,
        right=0.1 of Ts2,
        anchor=west,
    ] (Ts3) {$\hat{t}_3$};
    \node[
        below=-0.2 of Ts3,
    ] (Ts3b) {\tiny \texttt{life}};

    \node[
        fignode,
        right=0.1 of Ts3,
        anchor=west,
    ] (Ts4) {$\hat{t}_4$};
    \node[
        below=-0.2 of Ts4,
    ] (Ts4b) {\tiny \texttt{[SEP]}};

    % encoder
    \node[
        framedfignode,
        minimum width=150,
        above=1.25 of Ts1.west,
        anchor=west,
    ] (L) {Transformer encoder layer};

    % output
    \node[
        fignode,
        above=1.25 of L -| T1,
        anchor=center,
    ] (O) {$\hat{\eta}(d)$};

    % batch visualization
    \begin{scope}[on background layer]
        \node[
            xshift=4,
            yshift=-4,
            fit=(T1b.south west |- T1.north west)(Tdb.south east),
            inner sep=0,
            dashed,
            rounded corners,
            draw=gray,
        ] (Tr) {};
        \node[
            xshift=2,
            yshift=-2,
            fit=(T1b.south west |- T1.north west)(Tdb.south east),
            inner sep=0,
            dashed,
            rounded corners,
            fill=white,
            draw=gray,
        ] {};
        \node[
            fit=(T1b.south west |- T1.north west)(Tdb.south east),
            inner sep=0,
            dashed,
            rounded corners,
            fill=white,
            draw,
        ] {};

        \node[
            xshift=4,
            yshift=-4,
            fit=(Ts1b.south west |- Ts1.north west)(Ts4b.south east),
            inner sep=0,
            dashed,
            rounded corners,
            draw=gray,
        ] (Tsr) {};
        \node[
            xshift=2,
            yshift=-2,
            fit=(Ts1b.south west |- Ts1.north west)(Ts4b.south east),
            inner sep=0,
            dashed,
            rounded corners,
            fill=white,
            draw=gray,
        ] (Tsm) {};
        \node[
            fit=(Ts1b.south west |- Ts1.north west)(Ts4b.south east),
            inner sep=0,
            dashed,
            rounded corners,
            fill=white,
            draw,
        ] {};

        \node[
            xshift=4,
            yshift=-4,
            fit=(O),
            inner sep=0,
            dashed,
            rounded corners,
            draw=gray,
        ] {};
        \node[
            xshift=2,
            yshift=-2,
            fit=(O),
            inner sep=0,
            dashed,
            rounded corners,
            fill=white,
            draw=gray,
        ] (Om) {};
        \node[
            fit=(O),
            inner sep=0,
            dashed,
            rounded corners,
            fill=white,
            draw,
        ] {};
    \end{scope}

    % inputs -> embedding
    \draw[->] (T1) -- (T1 |- E.south);
    \draw[->] (T2) -- (T2 |- E.south);
    \draw[->] (T3) -- (T3 |- E.south);
    \draw[->] (T4) -- (T4 |- E.south);
    \draw[->] (Td) -- (Td |- E.south);

    % embedding -> scoring network
    \draw[->] (E.north -| T1) -- (S.south -| T1);
    \draw[->] (E.north -| T2) -- (S.south -| T2);
    \draw[->] (E.north -| T3) -- (S.south -| T3);
    \draw[->] (E.north -| T4) -- (S.south -| T4);
    \draw[->] (E.north -| Td) -- (S.south -| Td);

    % scoring network -> intermediate representations
    \draw[->] (S.north -| Ts1) -- (Ts1b |- Tsm.south);
    \draw[->] (S.north -| Ts2) -- (Ts2b |- Tsm.south);
    \draw[->] (S.north -| Ts3) -- (Ts3b |- Tsm.south);
    \draw[->] (S.north -| Ts4) -- (Ts4b |- Tsm.south);

    % intermediate representations -> encoder
    \draw[->] (Ts1) -- (L.south -| Ts1);
    \draw[->] (Ts2) -- (L.south -| Ts2);
    \draw[->] (Ts3) -- (L.south -| Ts3);
    \draw[->] (Ts4) -- (L.south -| Ts4);

    % encoder -> output
    \draw[->] (L.north -| T1) -- (Om.south -| T1);

    % annotations
    \draw [
        decorate,
        decoration={
                brace,
                raise=10,
                amplitude=5,
            },
    ] (T1b.south west |- Tr.south) -- (T1.north west -| T1b.south west) node[left=0.75,midway] {Input batch};

    \draw [
        decorate,
        decoration={
                brace,
                raise=10,
                amplitude=5,
            },
    ] (Ts1b.south west |- Tsr.south) -- (Ts1.north west -| Ts1b.south west) node[left=0.75,midway] {Shortened batch};

    \draw [
        decorate,
        decoration={
                brace,
                raise=10,
                amplitude=5,
            },
    ] (L.south west) -- (L.north west) node[left=0.75,midway] {$L \times$};
\end{tikzpicture}
    \caption{The fine-tuning and inference phase of \selbert{} document encoders. In the given example, the documents in the input batch are dynamically shortened to four tokens each based on the corresponding relevance scores. Note that the positional encoding that is added to BERT input tokens has been omitted in this figure.}
    \label{fig:efficient_encoders.doc.selbert}
\end{figure}
We refer to this approach as \selbert{}. It uses a \emph{scoring network} $\Phi: \mathbb{N} \mapsto [0, 1]$ to determine the relevance of each input token before feeding it into the encoding BERT model $\Psi$. We denote the parameters of the scoring network as $\theta_\Phi$ and the parameters of the BERT model as $\theta_\Psi$. We use a lightweight, non-contextual scoring network with three \num{384}-dimensional feed-forward layers and ReLU activations. The final layer outputs a scalar that is fed into a sigmoid activation function to compute the final score. \selbert{} models are trained in two steps.

\subsubsection{Pre-Training}
\label{sec:efficient_encoders.doc.pretraining}
The \textbf{first step} pre-trains the scoring network. $\theta_\Psi$ is initialized using the weights of a pre-trained \bert{} model (e.g., \bertbase{}), and $\theta_\Phi$ is initialized randomly. The complete model is then trained for a single epoch using the same data as during the unsupervised \bert{} pre-training step~\cite{devlin2019bert}. The scoring network $\Phi$ is taken into account by multiplying the embedding of an input token $t_i$ by its corresponding score, i.e.,
\begin{equation}
    x_i = E(t_i) \cdot \Phi(t_i) + P(t_i),
\end{equation}
where $E(t_i)$ is the token embedding and $P(t_i)$ is the positional encoding. The resulting representation $x_i$ is then used to compute self-attention in the first encoder layer.

In order to encourage the scoring network to output scores less than one, we introduce a regularization term using the $L_1$-norm over the scores, where $n$ is the input sequence length:
\begin{equation}
    \ell_1 = \sum_{i=0}^n \Phi(t_i).
\end{equation}
The final objective is a combination of the original BERT pre-training loss $\mathcal{L}$ and the scoring regularizer scaled by a hyperparameter $\lambda$:
\begin{equation}
    \min_{\theta_\Psi, \theta_\Phi} \left[ \mathcal{L}(\theta_\Psi, \theta_\Phi) + \lambda \cdot \ell_1(\theta_\Phi)\right].
\end{equation}

\subsubsection{Fine-Tuning and Inference}
\label{sec:efficient_encoders.doc.ft_inf}
The \textbf{second step}, referred to as \emph{fine-tuning}, only trains the BERT model $\Psi$, while the scoring network $\Phi$ remains frozen for the remainder of the training process. Furthermore, the weights of the BERT model obtained in the previous step, $\theta_\Psi$, are discarded and replaced by the same pre-trained model as before. The training objective during this stage is identical to that of other dual-encoder models (cf.\ \cref{sec:prelims.dual_encoders.training}).

During fine-tuning and inference (i.e., document encoding), we only retain the tokens with the highest scores; we set a ratio $p \in [0; 1]$ of the original input length to retain. As a result, the length of the input batch is shortened by $1-p$. This is achieved by removing the lowest scoring tokens from the input. Since individual documents within a batch are usually padded, $p$ always corresponds to the longest sequence in the batch. Consequently, padding tokens are always removed first before the scores of the other tokens are taken into account. The process is illustrated in \cref{fig:efficient_encoders.doc.selbert}.

\section{Experimental Setup}
\label{sec:setup}
In this section, we outline the experimental setup, including baselines, datasets, and further details about training and evaluation.

\subsection{Baselines}
\label{sec:setup.baselines}
We consider the following baselines:
\begin{enumerate}
    \item \textbf{Sparse retrievers} rely on term-based matching between queries and documents. We consider \bm{}, which uses term-based retrieval signals. \deepct{}~\cite{dai2020context}, \splade{}~\cite{formal2021splade}, and \spade{}~\cite{choi2022spade} use sparse representations, but contextualize terms in some fashion.
    \item \textbf{Dense retrievers} retrieve documents that are semantically similar to the query in a common embedding space. We consider \tct{}~\cite{lin2021batch}, \ance{}~\cite{xiong2021approximate}, and the more recent \aggr{}~\cite{lin2023aggretriever}. All three approaches are based on BERT encoders. Large documents are split into passages before indexing (maxP). These dense retrievers use exact (brute-force) nearest neighbor search as opposed to approximate nearest neighbor (ANN) search. We evaluate these methods in both the retrieval and re-ranking setting.
    \item \textbf{Hybrid retrievers} interpolate sparse and dense retriever scores. We consider \clear{}~\cite{gao2020complement}, a retrieval model that complements lexical models with semantic matching. Additionally, we consider the hybrid strategy described in Section~\ref{sec:prelims.hybrid_retrieval} as a baseline, using the dense retrievers above.
    \item \textbf{Re-rankers} operate on the documents retrieved by a sparse retriever (e.g., \bm{}). Each query-document pair is input into the re-ranker, which outputs a corresponding score. In this paper, we use a \bertcls{} re-ranker, where the output corresponding to the classification token is used as the score. Note that re-ranking is performed using the full documents (i.e., documents are not split into passages). If an input exceeds \num{512} tokens, it is truncated. Furthermore, we consider \tildetwo{}~\cite{zhuang2021fast} with \tildeone{} expansion.
\end{enumerate}

\subsection{Datasets}
\label{sec:setup.datasets}
We evaluate our models and baselines on a variety of diverse retrieval datasets:
\begin{enumerate}
    \item The \textbf{TREC Deep Learning track}~\cite{craswell2021trec} provides test sets and relevance judgments for retrieval and ranking evaluation on the MS MARCO corpora~\cite{nguyen2016ms}. We use both the passage and document ranking test sets from the years 2019 and 2020 for our experiments. In addition, we use the MS MARCO development sets to determine the optimal values for hyperparameters.
    \item The \textbf{BEIR benchmark}~\cite{thakur2021beir} is a collection of various IR datasets, which are commonly evaluated in a \emph{zero-shot} fashion, i.e., without using any of the data for training the model. We evaluate our models on a subset of the BEIR datasets, including tasks such as passage retrieval, question answering, and fact checking.
\end{enumerate}

\subsection{Evaluation Details}
\label{sec:setup.evaluation}
Our ranking experiments are performed on a single machine using an Intel Xeon Silver 4210 CPU and an NVIDIA Tesla V100 GPU. In our initial experiments (\cref{tab:results.rerank.model_eval_doc,tab:results.rerank.model_eval_passage}), we measured the per-query latency by performing each experiment four times and reporting the average latency, excluding the first measurement. In subsequent experiments (\cref{tab:results.rerank.first_stage,fig:results.query_enc.psg_results.latency_dev,fig:results.indexing_efficiency.sel_bert.latency}), we adjusted our way of measuring; we perform multiple runs of each experiment, where each run contains multiple latency measurements. We then report the average over all measurements of the fastest run. In~\cref{tab:results.rerank.model_eval_doc,tab:results.rerank.model_eval_passage}, latency is reported as the sum of scoring (this includes operations like encoding queries and documents, obtaining representations from a \fastforward{} index, computing the scores as dot-products, and so on), interpolation (cf.\ \cref{eqn:prelims.interpolated_reranking.interpolation}), and sorting cost. Any pre-processing or tokenization cost is ignored. Where applicable, dense models use a batch size of \num{256}. The first-stage (sparse) retrieval step is not included, as it is constant for all methods. The \fastforward{} indexes are loaded into the main memory entirely before they are accessed. In~\cref{tab:results.rerank.first_stage}, we report end-to-end latency, which includes retrieval, re-ranking, and tokenization cost.

\begin{table}
    \center
    \begin{tabular}{lll}
        \toprule
                                        & MS MARCO (documents)                                & MS MARCO (passages)                                             \\
        \midrule
        \multirowcell{2}[0pt][l]{\ance} & \scriptsize\texttt{castorini/ance-msmarco-doc-maxp} & \scriptsize\texttt{castorini/ance-msmarco-passage}              \\
                                        & \scriptsize\texttt{msmarco-doc-ance-maxp-bf}        & \scriptsize\texttt{msmarco-passage-ance-bf}                     \\
        \midrule
        \multirowcell{2}[0pt][l]{\tct}  & \scriptsize\texttt{castorini/tct\_colbert-msmarco}  & \scriptsize\texttt{castorini/tct\_colbert-msmarco}              \\
                                        & \scriptsize\texttt{msmarco-doc-tct\_colbert-bf}     & \scriptsize\texttt{msmarco-passage-tct\_colbert-bf}             \\
        \midrule
        \multirowcell{2}[0pt][l]{\aggr} & \scriptsize\multirowcell{2}[0pt][l]{-}              & \scriptsize\texttt{castorini/aggretriever-cocondenser}          \\
                                        &                                                     & \scriptsize\texttt{msmarco-v1-passage.aggretriever-cocondenser} \\
        \bottomrule
    \end{tabular}
    \caption{The pre-trained dense encoders and corresponding indexes we used in our experiments. In each cell, the first line corresponds to a pre-trained encoder (to be obtained from the HuggingFace Hub) and the second line is a pre-built index provided by \pyserini{}.}
    \label{tab:setup.evaluation.encoders_indexes}
\end{table}
We use the \pyserini{}~\cite{lin2021pyserini} toolkit, which provides a number of pre-trained encoders (available on the \emph{HuggingFace Hub}\footnote{\url{https://huggingface.co/models}}) and corresponding indexes (see \cref{tab:setup.evaluation.encoders_indexes}), for our retrieval experiments. Dense encoders (\ance{}, \tct{}, and \aggr{}) output \num{768}-dimensional representations. The sparse \bm{} retriever is provided by \pyserini{} as well. We use the pre-built indexes \texttt{msmarco-passage} ($k_1=0.82$, $b=0.68$) and \texttt{msmarco-doc} ($k_1=4.46$, $b=0.82$). Furthermore, we use \pyserini{} to run \splade{} with the provided \texttt{msmarco-passage-distill-splade-max} index and the pre-trained \dsplade{} model.

We use the MS MARCO development set to determine the interpolation parameter $\alpha$. We set $\alpha = 0.2$ for \tct{}, $\alpha = 0.5$ for \ance{}, and $\alpha = 0.7$ for \bertcls{} (\cref{sec:results.rerank}). For \aggr{}, we set $\alpha = 0.3$ for \bm{} re-ranking and $\alpha = 0.1$ for \splade{} re-ranking. For the dual-encoder models we trained ourselves (\cref{sec:results.query_enc,sec:results.index_size,sec:results.indexing_efficiency}), the value for $\alpha$ is determined based on nDCG@10 re-ranking results on the MS MARCO development set and varies slightly for each model.

\subsection{Training Details}
\label{sec:setup.training}
Our dual-encoder models are trained using the contrastive loss in \cref{eq:prelims.dual_encoders.training.loss}. For each training instance, we sample \num{8} hard negative documents using BM25. Additionally, we use in-batch negatives and a batch size of \num{4}, resulting in $|D^-| = 32$ negatives for each query. Each model is trained on four NVIDIA A100 GPUs. We set the learning rate to $1 \cdot 10^{-5}$ and use gradient accumulation of \num{32} batches (this results in an effective batch size of $4 \cdot 4 \cdot 32 = 512$). During training, we perform validation on the MS MARCO development set. Our models are trained until the average precision stops improving for five consecutive iterations. We exclusively train on the MS MARCO passage ranking corpus; the resulting models are then evaluated on multiple datasets (i.e., for BEIR, we do zero-shot evaluation). Our \selbert{} model (cf.\ \cref{sec:efficient_encoders.doc}) uses $\lambda = 10^{-6}$ during pre-training. We implemented our models and training pipeline using PyTorch,\footnote{\url{https://pytorch.org/}} PyTorch-Lightning,\footnote{\url{https://pytorchlightning.ai/}} and Transformers.\footnote{\url{https://huggingface.co/}}

\subsubsection{Dual-Encoder Architecture}
\label{sec:setup.training.architecture}
Our dual-encoder rankers consist of a query encoder $\zeta$ and a document encoder $\eta$ (cf.\ \ref{sec:prelims.dual_encoders}):
\begin{align}
    \zeta(q) & = || \mat{W}_\zeta \hat{\zeta}(q) + b_\zeta ||_2, \\
    \eta(d)  & = || \mat{W}_\eta \hat{\eta}(d) + b_\eta ||_2.
\end{align}
The models $\hat{\zeta}$ and $\hat{\eta}$ map queries and documents to arbitrary vector representations; examples for these models are pre-trained Transformers or the encoders described in \cref{sec:efficient_encoders}. We include optional trainable linear layers (with corresponding weights $\mat{W}_\zeta \in \mathbb{R}^{a \times d_\zeta}$, $\mat{W}_\eta \in \mathbb{R}^{a \times d_\eta}$, $b_\zeta \in \mathbb{R}^a$ and $b_\eta \in \mathbb{R}^a$) for heterogeneous encoders, where the dimensions of the representation vectors, $d_\zeta$ and $d_\eta$, do not match. We further $L_2$-normalize the representations during training and indexing; we do not normalize the query representations during ranking, as this would only scale the scores, but not change the final ranking.

\section{Results}
\label{sec:results}
In this section, we perform experiments to show the effectiveness and efficiency of \fastforward{} indexes. Each subsection corresponds to one of our research questions.

\subsection{How suitable are dual-encoder models for interpolation-based re-ranking in terms of performance and efficiency?}
\label{sec:results.rerank}
This section focuses on the effectiveness and efficiency of \fastforward{} indexes for re-ranking. We use pre-trained dual-encoders that are homogeneous (i.e., both encoders are identical models) for our experiments.

\subsubsection{Interpolation-based Re-Ranking Performance of Dual-Encoder Models}
\begin{table*}
    \centering
    \small
    \begin{tabular}{llllllllll}
        \toprule
                            & \multicolumn{3}{c}{\trecdldocn}
                            & \multicolumn{3}{c}{\trecdldoct}
                            & \multicolumn{3}{c}{\trecdlpsgn}                                                   \\
        \cmidrule(lr){2-4}
        \cmidrule(lr){5-7}
        \cmidrule(lr){8-10}
                            & $\text{AP}_\text{1k}$           & $\text{R}_\text{1k}$ & $\text{nDCG}_\text{10}$
                            & $\text{AP}_\text{1k}$           & $\text{R}_\text{1k}$ & $\text{nDCG}_\text{10}$
                            & $\text{AP}_\text{1k}$           & $\text{R}_\text{1k}$ & $\text{nDCG}_\text{10}$  \\
        \midrule
        \multicolumn{10}{l}{\bfseries \sparseretrieval}                                                         \\
        \bm                 & \num{0.331}                     & \num{0.697}          & \num{0.519}\sigimpr{abc}
                            & \num{0.404}                     & \num{0.809}          & \num{0.527}\sigimpr{abc}
                            & \num{0.301}                     & \num{0.750}          & \num{0.506}\sigimpr{abc} \\
        \deepct             & -                               & -                    & \num{0.544}
                            & -                               & -                    & -
                            & \num{0.422}                     & \num{0.756}          & \num{0.551}              \\
        \midrule
        \multicolumn{10}{l}{\bfseries \denseretrieval}                                                          \\
        \tct                & \num{0.279}                     & \num{0.576}          & \num{0.612}\sigimpr{a}
                            & \num{0.372}                     & \num{0.728}          & \num{0.586}\sigimpr{ab}
                            & \num{0.391}                     & \num{0.792}          & \num{0.670}              \\
        \ance               & \num{0.254}                     & \num{0.510}          & \num{0.633}\sigimpr{a}
                            & \num{0.401}                     & \num{0.681}          & \num{0.633}
                            & \num{0.371}                     & \num{0.755}          & \num{0.645}              \\
        \midrule
        \multicolumn{10}{l}{\bfseries \hybrid}                                                                  \\
        \clear              & -                               & -                    & -
                            & -                               & -                    & -
                            & \num{0.511}                     & \num{0.812}          & \num{0.699}              \\
        \midrule
        \multicolumn{10}{l}{\bfseries \reranking}                                                               \\
        \tct                & \num{0.370}                     & \num{0.697}          & \num{0.685}
                            & \num{0.414}                     & \num{0.809}          & \num{0.617}
                            & \num{0.423}                     & \num{0.750}          & \num{0.694}              \\
        \ance               & \num{0.336}                     & \num{0.697}          & \num{0.654}
                            & \num{0.426}                     & \num{0.809}          & \num{0.630}
                            & \num{0.389}                     & \num{0.750}          & \num{0.679}              \\
        \bertcls            & \num{0.283}                     & \num{0.697}          & \num{0.520}\sigimpr{abc}
                            & \num{0.329}                     & \num{0.809}          & \num{0.522}\sigimpr{abc}
                            & \num{0.353}                     & \num{0.750}          & \num{0.578}\sigimpr{ab}  \\
        \midrule
        \multicolumn{10}{l}{\bfseries \interpolatedreranking}                                                   \\
        \sigdef{a} \tct     & \num{0.406}                     & \num{0.697}          & \num{0.696}
                            & \num{0.469}                     & \num{0.809}          & \num{0.637}
                            & \num{0.438}                     & \num{0.750}          & \num{0.708}              \\
        \sigdef{b} \ance    & \num{0.387}                     & \num{0.697}          & \num{0.673}
                            & \num{0.490}                     & \num{0.809}          & \num{0.655}
                            & \num{0.417}                     & \num{0.750}          & \num{0.680}              \\
        \sigdef{c} \bertcls & \num{0.365}                     & \num{0.697}          & \num{0.612}
                            & \num{0.460}                     & \num{0.809}          & \num{0.626}
                            & \num{0.378}                     & \num{0.750}          & \num{0.617}              \\
        \bottomrule
    \end{tabular}
    \caption{Ranking performance. Retrievers use depths $k_S = 1000$ (sparse) and $k_D = 10000$ (dense). Dense retrievers retrieve passages and perform maxP aggregation for documents. Scores for \clear{} and \deepct{} are taken from the corresponding papers~\cite{gao2020complement,gao2021coil}. Superscripts indicate statistically significant improvements using two-paired tests with a sig.\ level of \num{95}\%~\cite{paired_significance_test}.}
    \label{tab:results.rerank.model_eval}
\end{table*}
In \cref{tab:results.rerank.model_eval}, we report the performance of sparse, dense and hybrid retrievers, re-rankers and interpolation.

First, we observe that dense retrieval strategies perform better than sparse ones in terms of nDCG, but have poor recall except on \trecdlpsgn{}. The contextual weights learned by \deepct{} are better than tf-idf-based retrieval (\bm{}), but fall short of dense semantic retrieval strategies (\tct{} and \ance{}) with differences upwards of \num{0.1} in nDCG. However, the overlap among retrieved documents is rather low, reflecting that dense retrieval cannot match query and document terms well.

Second, dual-encoder-based (\tct{} and \ance{}) perform better than contextual (\bertcls{}) re-rankers. In this setup, we first retrieve $k_S = 1000$ documents using a sparse retriever and re-rank them. This approach benefits from high recall in the first stage and promotes the relevant documents to the top of the list through the dense semantic re-ranker. However, re-ranking is typically time-consuming and requires GPU acceleration. The improvements of \tct{} and \ance{} over \bertcls{} (e.g., \num{0.1} in nDCG) also suggest that dual-encoder-based re-ranking strategies are better than cross-interaction-based methods. However, the difference could also be attributed to the fact that \bertcls{} does not follow the maxP approach (cf.\ \cref{sec:prelims.interpolated_reranking}).

Finally, interpolation-based re-ranking, which combines the benefits of sparse and dense scores, significantly outperforms the \bertcls{} re-ranker and dense retrievers. Recall that dense re-rankers operate solely based on the dense scores and discard the sparse \bm{} scores of the query-document pairs. The superiority of interpolation-based methods is also supported by evidence from recent studies~\cite{chang2020pre,chen2021co,gao2020complement,gao2021coil}.

\subsubsection{Efficient Re-Ranking at Higher Retrieval Depths}
\begin{sidewaystable}
    \small
    \centering
    \begin{tabular}{lr@{\hspace{0.5\tabcolsep}}c@{\hspace{0.5\tabcolsep}}lllllllllllll}
        \toprule
         &
         &
         &
         & \multicolumn{6}{c}{\trecdldocn}
         & \multicolumn{6}{c}{\trecdldoct}
        \\
        \cmidrule(lr){5-10}
        \cmidrule(lr){11-16}
         & \multicolumn{3}{c}{Latency}
         & \multicolumn{3}{c}{$k_S = 1000$}
         & \multicolumn{3}{c}{$k_S = 5000$}
         & \multicolumn{3}{c}{$k_S = 1000$}
         & \multicolumn{3}{c}{$k_S = 5000$}
        \\
        \cmidrule(lr){2-4}
        \cmidrule(lr){5-7}
        \cmidrule(lr){8-10}
        \cmidrule(lr){11-13}
        \cmidrule(lr){14-16}
         & \multicolumn{3}{c}{ms}
         & $\text{AP}_\text{1k}$            & $\text{R}_\text{1k}$ & $\text{nDCG}_\text{20}$
         & $\text{AP}_\text{1k}$            & $\text{R}_\text{1k}$ & $\text{nDCG}_\text{20}$
         & $\text{AP}_\text{1k}$            & $\text{R}_\text{1k}$ & $\text{nDCG}_\text{20}$
         & $\text{AP}_\text{1k}$            & $\text{R}_\text{1k}$ & $\text{nDCG}_\text{20}$
        \\
        \midrule
        \multicolumn{16}{l}{\bfseries \hybrid}                                                                                                                         \\
        \bm, \tct
         & \gpu{\num{0}}                    & +                    & \cpu{\num{582}}
         & \num{0.394}                      & \num{0.697}          & \num{0.655}                & \num{0.385}        & \num{0.729}        & \num{0.645}
         & \num{0.463}                      & \num{0.809}          & \num{0.615}                & \num{0.469}        & \num{0.852}        & \num{0.621}
        \\
        \bm, \ance
         & \gpu{\num{0}}                    & +                    & \cpu{\num{582}}
         & \num{0.379}                      & \num{0.697}          & \num{0.633}                & \num{0.373}        & \num{0.727}        & \num{0.628}
         & \num{0.479}                      & \num{0.809}          & \num{0.624}                & \num{0.488}        & \num{0.846}        & \num{0.632}
        \\
        \midrule
        \multicolumn{16}{l}{\bfseries \reranking}                                                                                                                      \\
        \tct
         & \gpu{\num{1189}}                 & +                    & \cpu{2}
         & \num{0.370}                      & \num{0.697}          & \num{0.632}                & \num{0.334}        & \num{0.703}        & \num{0.609}\sigimpr{a}
         & \num{0.414}                      & \num{0.809}          & \num{0.587}\sigimpr{a}     & \num{0.405}        & \num{0.794}        & \num{0.585}\sigimpr{acd}
        \\
        \ance
         & \gpu{\num{1189}}                 & +                    & \cpu{\num{2}}
         & \num{0.336}                      & \num{0.697}          & \num{0.614}                & \num{0.304}        & \num{0.647}        & \num{0.607}
         & \num{0.426}                      & \num{0.809}          & \num{0.595}\sigimpr{c}     & \num{0.422}        & \num{0.761}        & \num{0.604}
        \\
        \bertcls
         & \gpu{\num{185}}                  & +                    & \cpu{\num{2}}
         & \num{0.283}                      & \num{0.697}          & \num{0.494}\sigimpr{abcde} & \num{0.159}        & \num{0.559}        & \num{0.289}
         & \num{0.329}                      & \num{0.809}          & \num{0.512}\sigimpr{abcde} & \num{0.221}        & \num{0.727}        & \num{0.375}\sigimpr{abcde}
        \\
        \midrule
        \multicolumn{16}{l}{\bfseries \interpolatedreranking}                                                                                                          \\
        \midrulesep
        \hide{\sigdef{a}} \tct
         & \gpu{\num{1189}}                 & +                    & \cpu{\num{14}}
         & \hide{\num{0.406}}               & \hide{\num{0.697}}   & \hide{\num{0.655}}         & \hide{\num{0.411}} & \hide{\num{0.745}} & \hide{\num{0.653}}
         & \hide{\num{0.469}}               & \hide{\num{0.809}}   & \hide{\num{0.621}}         & \hide{\num{0.478}} & \hide{\num{0.838}} & \hide{\num{0.626}}
        \\
        \sigdef{a} \tablearrow \fastforward
         & \gpu{\num{0}}                    & +                    & \cpu{\num{253}}
         & \num{0.406}                      & \num{0.697}          & \num{0.655}                & \num{0.411}        & \num{0.745}        & \num{0.653}
         & \num{0.469}                      & \num{0.809}          & \num{0.621}                & \num{0.478}        & \num{0.838}        & \num{0.626}
        \\
        \sigdef{b} \quad \tablearrow coalesced
         & \gpu{\num{0}}                    & +                    & \cpu{\num{109}}
         & \num{0.379}                      & \num{0.697}          & \num{0.630}                & \num{0.379}        & \num{0.732}        & \num{0.625}
         & \num{0.440}                      & \num{0.809}          & \num{0.59}4\sigimpr{a}     & \num{0.447}        & \num{0.837}        & \num{0.607}
        \\
        \midrulesep
        \hide{\sigdef{c}} \ance
         & \gpu{\num{1189}}                 & +                    & \cpu{\num{14}}
         & \hide{\num{0.387}}               & \hide{\num{0.697}}   & \hide{\num{0.638}}         & \hide{\num{0.393}} & \hide{\num{0.732}} & \hide{\num{0.639}}
         & \hide{\num{0.490}}               & \hide{\num{0.809}}   & \hide{\num{0.630}}         & \hide{\num{0.502}} & \hide{\num{0.828}} & \hide{\num{0.640}}
        \\
        \sigdef{c} \tablearrow \fastforward
         & \gpu{\num{0}}                    & +                    & \cpu{\num{253}}
         & \num{0.387}                      & \num{0.697}          & \num{0.638}                & \num{0.393}        & \num{0.732}        & \num{0.639}
         & \num{0.490}                      & \num{0.809}          & \num{0.630}                & \num{0.502}        & \num{0.828}        & \num{0.640}
        \\
        \sigdef{d} \quad \tablearrow coalesced
         & \gpu{\num{0}}                    & +                    & \cpu{\num{121}}
         & \num{0.372}                      & \num{0.697}          & \num{0.625}                & \num{0.375}        & \num{0.723}        & \num{0.628}
         & \num{0.471}                      & \num{0.809}          & \num{0.622}                & \num{0.479}        & \num{0.823}        & \num{0.629}
        \\
        \midrulesep
        \sigdef{e} \bertcls
         & \gpu{\num{185}}                  & +                    & \cpu{\num{14}}
         & \num{0.365}                      & \num{0.697}          & \num{0.585}                & \num{0.357}        & \num{0.708}        & \num{0.562}
         & \num{0.460}                      & \num{0.809}          & \num{0.602}                & \num{0.459}        & \num{0.839}        & \num{0.601}
        \\
        \bottomrule
    \end{tabular}
    \caption{Document ranking performance. Latency is reported per query for $k_S = 5000$ on \gpu{GPU} and \cpu{CPU}. The coalesced \fastforward{} indexes are compressed to approximately 25\% of their original size. Hybrid retrievers use a dense retrieval depth of $k_D = 1000$. Superscripts indicate statistically significant improvements using two-paired tests with a sig.\ level of \num{95}\%~\cite{paired_significance_test}.}
    \label{tab:results.rerank.model_eval_doc}
\end{sidewaystable}
\begin{table*}
    \centering
    \begin{tabular}{lr@{\hspace{0.5\tabcolsep}}c@{\hspace{0.5\tabcolsep}}lcccccc}
        \toprule
         & \multicolumn{3}{c}{Latency}
         & \multicolumn{2}{c}{$k_S = 1000$}
         & \multicolumn{2}{c}{$k_S = 5000$}
        \\
        \cmidrule(lr){2-4}
        \cmidrule(lr){5-6}
        \cmidrule(lr){7-8}
         & \multicolumn{3}{c}{ms}
         & $\text{AP}_\text{1k}$            & $\text{RR}_\text{10}$
         & $\text{AP}_\text{1k}$            & $\text{RR}_\text{10}$
        \\
        \midrule
        \multicolumn{8}{l}{\bfseries \hybrid}                                                                 \\
        \bm, \tct
         & \gpu{\num{0}}                    & +                     & \cpu{\num{307}}
         & \num{0.434}                      & \num{0.894}           & \num{0.454}        & \num{0.902}
        \\
        \bm, \ance
         & \gpu{\num{0}}                    & +                     & \cpu{\num{307}}
         & \num{0.410}                      & \num{0.856}           & \num{0.422}        & \num{0.864}
        \\
        \midrule
        \multicolumn{8}{l}{\bfseries \reranking}                                                              \\
        \tct
         & \gpu{\num{186}}                  & +                     & \cpu{\num{2}}
         & \num{0.426}                      & \num{0.827}           & \num{0.439}        & \num{0.842}
        \\
        \ance
         & \gpu{\num{186}}                  & +                     & \cpu{\num{2}}
         & \num{0.389}                      & \num{0.836}           & \num{0.392}        & \num{0.857}
        \\
        \bertcls
         & \gpu{\num{185}}                  & +                     & \cpu{\num{2}}
         & \num{0.353}                      & \num{0.715}           & \num{0.275}        & \num{0.576}
        \\
        \midrule
        \multicolumn{8}{l}{\bfseries \interpolatedreranking}                                                  \\
        \midrulesep
        \tct
         & \gpu{\num{186}}                  & +                     & \cpu{\num{14}}
         & \hide{\num{0.438}}               & \hide{\num{0.894}}    & \hide{\num{0.460}} & \hide{\num{0.902}}
        \\
        \tablearrow \fastforward
         & \gpu{\num{0}}                    & +                     & \cpu{\num{114}}
         & \num{0.438}                      & \num{0.894}           & \num{0.460}        & \num{0.902}
        \\
        \quad \tablearrow early stopping
         & \gpu{\num{0}}                    & +                     & \cpu{\num{72}}
         & -                                & \num{0.894}           & -                  & \num{0.902}
        \\
        \midrulesep
        \ance
         & \gpu{\num{186}}                  & +                     & \cpu{\num{14}}
         & \hide{\num{0.417}}               & \hide{\num{0.856}}    & \hide{\num{0.435}} & \hide{\num{0.864}}
        \\
        \tablearrow \fastforward
         & \gpu{\num{0}}                    & +                     & \cpu{\num{114}}
         & \num{0.417}                      & \num{0.856}           & \num{0.435}        & \num{0.864}
        \\
        \quad \tablearrow early stopping
         & \gpu{\num{0}}                    & +                     & \cpu{\num{52}}
         & -                                & \num{0.856}           & -                  & \num{0.864}
        \\
        \midrulesep
        \bertcls
         & \gpu{\num{185}}                  & +                     & \cpu{\num{14}}
         & \num{0.378}                      & \num{0.809}           & \num{0.392}        & \num{0.832}
        \\
        \bottomrule
    \end{tabular}
    \caption{Ranking performance on \trecdlpsgn{}. Latency is reported per query for $k_S = 5000$ on \gpu{GPU} and \cpu{CPU}. Hybrid retrievers use a dense retrieval depth of $k_D = 1000$.}
    \label{tab:results.rerank.model_eval_passage}
\end{table*}
\Cref{tab:results.rerank.model_eval_doc,tab:results.rerank.model_eval_passage} show results of re-ranking, hybrid retrieval and interpolation on document and passage datasets, respectively. The metrics are computed for two sparse retrieval depths, $k_S = 1000$ and $k_S = 5000$.

We observe that additionally taking the sparse component into account in the score computation (as is done by the interpolation and hybrid methods) causes performance to improve with retrieval depth. Specifically, some queries receive a considerable recall boost, capturing more relevant documents with large retrieval depths. Interpolation based on \fastforward{} indexes achieves substantially lower latency compared to other methods. Pre-computing the document representations allows for fast look-ups during retrieval time. As only the query needs to be encoded by the dense model, both retrieval and re-ranking can be performed on the CPU while still offering considerable improvements in query processing time. Note that for \bertcls{}, the input length is limited, causing documents to be truncated, similarly to the \emph{firstP} approach. As a result, the latency is much lower, but in turn the performance suffers. It is important to note here, that, in principle, \fastforward{} indexes can also be used in combination with firstP models.

The hybrid retrieval strategy, as described in \cref{sec:prelims.hybrid_retrieval}, shows good performance. However, as the dense indexes require nearest neighbor search for retrieval, the query processing latency is much higher than for interpolation using \fastforward{} indexes.

Finally, dense re-rankers do not profit reliably from increased sparse retrieval depth; on the contrary, the performance drops in some cases. This trend is more apparent for the document retrieval datasets with higher values of $k_S$. We hypothesize that dense rankers only focus on semantic matching and are sensitive to topic drift, causing them to rank irrelevant documents in the top-\num{5000} higher.

\subsubsection{Varying the First-Stage Retrieval Model}
\begin{table*}
    \centering
    \small
    \begin{tabular}{lccccccccc}
        \toprule
                                       & Latency
                                       & \multicolumn{2}{c}{\msmpsgdev}
                                       & \multicolumn{3}{c}{\trecdlpsgn}
                                       & \multicolumn{3}{c}{\trecdlpsgt}                                                   \\
        \cmidrule(lr){2-2}
        \cmidrule(lr){3-4}
        \cmidrule(lr){5-7}
        \cmidrule(lr){8-10}
                                       & ms
                                       & $\text{AP}_\text{1k}$           & $\text{RR}_\text{10}$
                                       & $\text{AP}_\text{1k}$           & $\text{RR}_\text{10}$ & $\text{nDCG}_\text{10}$
                                       & $\text{AP}_\text{1k}$           & $\text{RR}_\text{10}$ & $\text{nDCG}_\text{10}$ \\
        \midrule
        \bm                            & \cpu{\num{14}}
                                       & \num{0.196}                     & \num{0.187}
                                       & \num{0.301}                     & \num{0.702}           & \num{0.506}
                                       & \num{0.288}                     & \num{0.655}           & \num{0.488}             \\
        \tablearrow \tildetwo          & \cpu{104}
                                       & \num{0.338}                     & \num{0.342}
                                       & \num{0.437}                     & \num{0.836}           & \num{0.680}
                                       & \num{0.459}                     & \num{0.868}           & \num{0.679}             \\
        \tablearrow \aggr                                                                                                  \\
        \quad \tablearrow \fastforward & \cpu{\num{150}}
                                       & \num{0.373}                     & \num{0.369}
                                       & \num{0.465}                     & \num{0.877}           & \num{0.700}
                                       & \num{0.486}                     & \num{0.825}           & \num{0.717}             \\
        \midrule
        \spade{} ($k=5$)               & -
                                       & -                               & \num{0.355}
                                       & \num{0.437}                     & -                     & \num{0.682}
                                       & \num{0.453}                     & -                     & \num{0.677}             \\
        \midrule
        \splade                        & \cpu{\num{302}}
                                       & \num{0.375}                     & \num{0.368}
                                       & \num{0.485}                     & \num{0.901}           & \num{0.728}
                                       & \num{0.490}                     & \num{0.830}           & \num{0.711}             \\
        \tablearrow \tildetwo          & \cpu{\num{374}}
                                       & \num{0.337}                     & \num{0.342}
                                       & \num{0.412}                     & \num{0.808}           & \num{0.654}
                                       & \num{0.433}                     & \num{0.858}           & \num{0.648}             \\
        \tablearrow \aggr                                                                                                  \\
        \quad \tablearrow \fastforward & \cpu{\num{420}}
                                       & \num{0.383}                     & \num{0.378}
                                       & \num{0.489}                     & \num{0.899}           & \num{0.726}
                                       & \num{0.500}                     & \num{0.856}           & \num{0.716}             \\
        \bottomrule
    \end{tabular}
    \caption{Passage ranking performance using various first-stage retrieval models as well as re-rankers. \aggr{} models are used for interpolation-based re-ranking using \fastforward{} indexes. Re-ranking is done with $k_S = 5000$ passages. \spade{} results are taken from the corresponding paper~\cite{choi2022spade}. For \splade{}, we use the \dsplade{} model. Latency is reported per query on \cpu{CPU}. For retrieval models (\bm{} and \splade{}), latency is reported at retrieval depth $k_S = 1000$. For re-ranking (\tildetwo{} and \fastforward{}), latency is reported as the sum of retrieval and re-ranking, both at depth $k_S = 5000$.}
    \label{tab:results.rerank.first_stage}
\end{table*}
We perform additional passage ranking experiments in \cref{tab:results.rerank.first_stage}, where we compare various first-stage retrieval methods in combination with re-rankers. The idea is to show how \fastforward{} indexes perform in combination with modern sparse retrievers and how they compare with other re-rankers. Additionally, these experiments give an idea of the \emph{end-to-end} efficiency, as we report the latency as the sum of retrieval, re-ranking, and tokenization. The \aggr{} model~\cite{lin2023aggretriever} we use in combination with \fastforward{} indexes is a recent single-vector dual-encoder model based on \cocondenser{}~\cite{gao2022unsupervised}.

Both \spade{} and \splade{}, unsurprisingly, perform substantially better than \bm{}, as these models use contextualized learnt representations. This boost in performance comes with a large increase in latency, in terms of both indexing and query processing. However, it becomes evident that re-ranking \bm{} results comes very close to these models in terms of performance, and sometimes even surpasses them, even though the overall latency remains lower. At the same time, \fastforward{} indexes manage to improve the performance of \splade{} by re-ranking (although the improvements are not as big). Interestingly, \tildetwo{} does not exhibit this behavior, but rather performs worse when a \splade{} first-stage retriever is used. We assume that the reason for this is that the model was not optimized for this scenario.

\subsection{Can the re-ranking efficiency be improved by limiting the number of \fastforward{} look-ups?}
\label{sec:results.lookups}
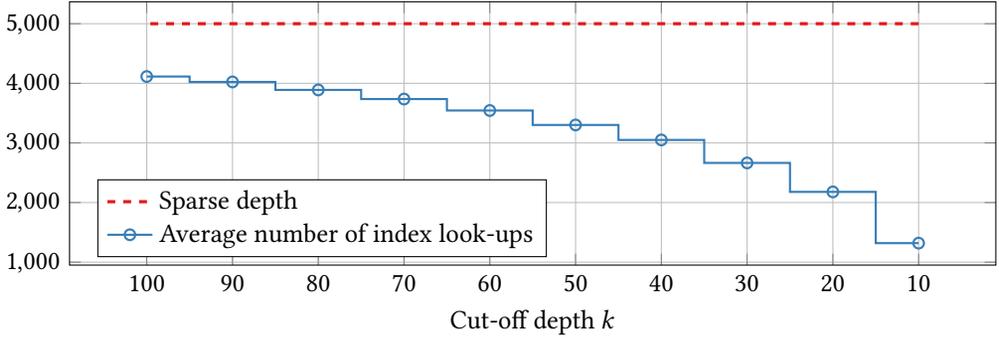
\begin{figure}
    \centering
    \begin{tikzpicture}
    \pgfplotstableread[col sep=comma]{plots/early_stopping/early_stopping.csv}\data
    \begin{axis}[
            width=\textwidth,
            height=0.25\textheight,
            xlabel={Cut-off depth $k$},
            grid=major,
            x dir=reverse,
            legend entries={
                    Sparse depth,
                    Average number of index look-ups
                },
            legend cell align={left},
            legend pos=south west,
        ]
        \addplot[
            domain=10:100,
            samples=2,
            dashed,
            mark=none,
            very thick,
            plotColor1,
        ] {5000};
        \addplot+[
            const plot mark mid,
            thick,
            mark=o,
            plotColor2,
        ] table[y index=1] {\data};
    \end{axis}
\end{tikzpicture}
    \caption{The average number of \fastforward{} index look-ups per query for interpolation with early stopping at varying cut-off depths $k$ on \trecdlpsgn{} with $k_S = 5000$ using \ance{}.}
    \label{fig:results.lookups.early_stopping}
\end{figure}
We evaluate the utility of the early stopping approach described in \cref{sec:ff_indexes.early_stopping} on the \trecdlpsgn{} dataset. \Cref{fig:results.lookups.early_stopping} shows the average number of look-ups performed in the \fastforward{} index during interpolation w.r.t.\ the cut-off depth $k$. We observe that, for $k = 100$, early stopping already leads to a reduction of almost \num{20}\% in the number of look-ups. Decreasing $k$ further leads to a significant reduction of look-ups, resulting in improved query processing latency. As lower cut-off depths (i.e., $k < 100$) are typically used in downstream tasks, such as question answering, the early stopping approach for low values of $k$ turns out to be particularly helpful.

\Cref{tab:results.rerank.model_eval_passage} shows early stopping applied to the passage dataset to retrieve the top-\num{10} passages and compute reciprocal rank. It is evident that, even though the algorithm approximates the maximum dense score (cf.\ \cref{sec:ff_indexes.early_stopping}), the resulting performance is identical, which means that the approximation was accurate in both cases and did not incur any performance hit. Furthermore, the query processing time is decreased by up to a half compared to standard interpolation. This means that presenting a small number top results (as is common in many downstream tasks) can yield substantial speed-ups. Note that early stopping depends on the value of $\alpha$, hence the latency varies between \tct{} and \ance{}.

\subsection{To what extent does query encoder complexity affect re-ranking performance?}
\label{sec:results.query_enc}
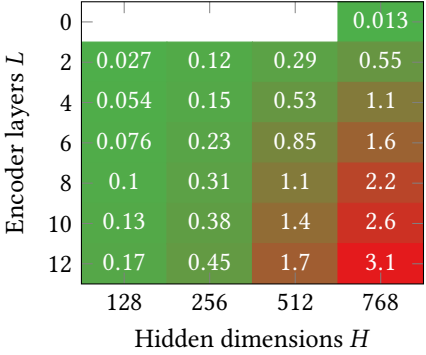
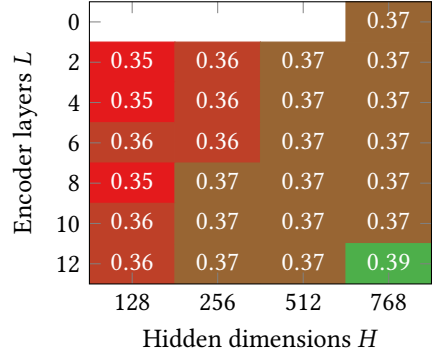
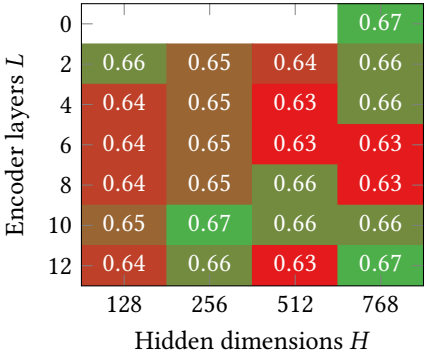
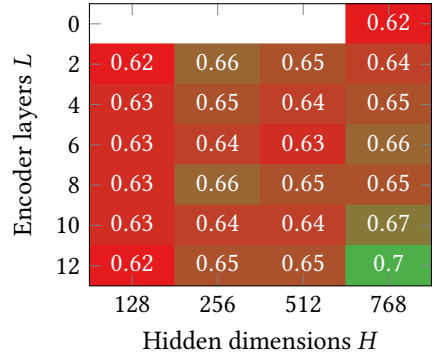
\begin{figure*}
  \begin{subfigure}{.49\linewidth}
    \centering
    \begin{tikzpicture}
    \begin{axis}[
            width=0.9\textwidth,
            enlargelimits=false,
            colormap name={greenred},
            xlabel={Hidden dimensions $H$},
            xtick={0,1,2,3},
            xticklabels={128,256,512,768},
            ylabel={Encoder layers $L$},
            ytick={0,1,2,3,4,5,6},
            yticklabels={0,2,4,6,8,10,12},
        ]
        \addplot [
            matrix plot,
            nodes near coords style={
                    anchor=mid,
                    text=white,
                },
            nodes near coords={
                    \ifthenelse{
                        \equal{\pgfplotspointmetatransformed}{0}}
                    {}
                    {\pgfmathprintnumber[precision=4]{\pgfplotspointmeta}}
                },
            point meta=explicit,
        ] coordinates {
                (0,0) [nan] (1,0) [nan] (2,0) [nan] (3,0) [0.013]

                (0,1) [0.027] (1,1) [0.12] (2,1) [0.29] (3,1) [0.55]

                (0,2) [0.054] (1,2) [0.15] (2,2) [0.53] (3,2) [1.1]

                (0,3) [0.076] (1,3) [0.23] (2,3) [0.85] (3,3) [1.6]

                (0,4) [0.1] (1,4) [0.31] (2,4) [1.1] (3,4) [2.2]

                (0,5) [0.13] (1,5) [0.38] (2,5) [1.4] (3,5) [2.6]

                (0,6) [0.17] (1,6) [0.45] (2,6) [1.7] (3,6) [3.1]
            };
    \end{axis}
\end{tikzpicture}
    \caption{Query encoding latency in seconds}
    \label{fig:results.query_enc.psg_results.latency_dev}
  \end{subfigure}
  \hfill
  \begin{subfigure}{.49\linewidth}
    \centering
    \begin{tikzpicture}
    \begin{axis}[
            width=0.9\textwidth,
            enlargelimits=false,
            colormap name={redgreen},
            xlabel={Hidden dimensions $H$},
            xtick={0,1,2,3},
            xticklabels={128,256,512,768},
            ylabel={Encoder layers $L$},
            ytick={0,1,2,3,4,5,6},
            yticklabels={0,2,4,6,8,10,12},
        ]
        \addplot [
            matrix plot,
            nodes near coords style={
                    anchor=mid,
                    text=white,
                },
            nodes near coords={
                    \ifthenelse{
                        \equal{\pgfplotspointmetatransformed}{0}}
                    {}
                    {\pgfmathprintnumber[precision=4]{\pgfplotspointmeta}}
                },
            point meta=explicit,
        ] coordinates {
                (0,0) [nan] (1,0) [nan] (2,0) [nan] (3,0) [0.37]

                (0,1) [0.35] (1,1) [0.36] (2,1) [0.37] (3,1) [0.37]

                (0,2) [0.35] (1,2) [0.36] (2,2) [0.37] (3,2) [0.37]

                (0,3) [0.36] (1,3) [0.36] (2,3) [0.37] (3,3) [0.37]

                (0,4) [0.35] (1,4) [0.37] (2,4) [0.37] (3,4) [0.37]

                (0,5) [0.36] (1,5) [0.37] (2,5) [0.37] (3,5) [0.37]

                (0,6) [0.36] (1,6) [0.37] (2,6) [0.37] (3,6) [0.39]
            };
    \end{axis}
\end{tikzpicture}
    \caption{nDCG@\num{10} on \msmpsgdev{}}
    \label{fig:results.query_enc.psg_results.ndcg_dev}
  \end{subfigure}
  \par\bigskip\bigskip
  \begin{subfigure}{.49\linewidth}
    \centering
    \begin{tikzpicture}
    \begin{axis}[
            width=0.9\textwidth,
            enlargelimits=false,
            colormap name={redgreen},
            xlabel={Hidden dimensions $H$},
            xtick={0,1,2,3},
            xticklabels={128,256,512,768},
            ylabel={Encoder layers $L$},
            ytick={0,1,2,3,4,5,6},
            yticklabels={0,2,4,6,8,10,12},
        ]
        \addplot [
            matrix plot,
            nodes near coords style={
                    anchor=mid,
                    text=white,
                },
            nodes near coords={
                    \ifthenelse{
                        \equal{\pgfplotspointmetatransformed}{0}}
                    {}
                    {\pgfmathprintnumber[precision=4]{\pgfplotspointmeta}}
                },
            point meta=explicit,
        ] coordinates {
                (0,0) [nan] (1,0) [nan] (2,0) [nan] (3,0) [0.67]

                (0,1) [0.66] (1,1) [0.65] (2,1) [0.64] (3,1) [0.66]

                (0,2) [0.64] (1,2) [0.65] (2,2) [0.63] (3,2) [0.66]

                (0,3) [0.64] (1,3) [0.65] (2,3) [0.63] (3,3) [0.63]

                (0,4) [0.64] (1,4) [0.65] (2,4) [0.66] (3,4) [0.63]

                (0,5) [0.65] (1,5) [0.67] (2,5) [0.66] (3,5) [0.66]

                (0,6) [0.64] (1,6) [0.66] (2,6) [0.63] (3,6) [0.67]
            };
    \end{axis}
\end{tikzpicture}
    \caption{nDCG@\num{10} on \trecdlpsgn{}}
    \label{fig:results.query_enc.psg_results.ndcg_psg19}
  \end{subfigure}
  \hfill
  \begin{subfigure}{.49\linewidth}
    \centering
    \begin{tikzpicture}
    \begin{axis}[
            width=0.9\textwidth,
            enlargelimits=false,
            colormap name={redgreen},
            xlabel={Hidden dimensions $H$},
            xtick={0,1,2,3},
            xticklabels={128,256,512,768},
            ylabel={Encoder layers $L$},
            ytick={0,1,2,3,4,5,6},
            yticklabels={0,2,4,6,8,10,12},
        ]
        \addplot [
            matrix plot,
            nodes near coords style={
                    anchor=mid,
                    text=white,
                },
            nodes near coords={
                    \ifthenelse{
                        \equal{\pgfplotspointmetatransformed}{0}}
                    {}
                    {\pgfmathprintnumber[precision=4]{\pgfplotspointmeta}}
                },
            point meta=explicit,
        ] coordinates {
                (0,0) [nan] (1,0) [nan] (2,0) [nan] (3,0) [0.62]

                (0,1) [0.62] (1,1) [0.66] (2,1) [0.65] (3,1) [0.64]

                (0,2) [0.63] (1,2) [0.65] (2,2) [0.64] (3,2) [0.65]

                (0,3) [0.63] (1,3) [0.64] (2,3) [0.63] (3,3) [0.66]

                (0,4) [0.63] (1,4) [0.66] (2,4) [0.65] (3,4) [0.65]

                (0,5) [0.63] (1,5) [0.64] (2,5) [0.64] (3,5) [0.67]

                (0,6) [0.62] (1,6) [0.65] (2,6) [0.65] (3,6) [0.7]
            };
    \end{axis}
\end{tikzpicture}
    \caption{nDCG@\num{10} on \trecdlpsgt{}}
    \label{fig:results.query_enc.psg_results.ndcg_psg20}
  \end{subfigure}
  \caption{Query encoding latency and \fastforward{} ranking performance of dual-encoders with various query encoder models. The sparse retrieval depth is $k_S = 5000$. $L$ and $H$ correspond to the number of Transformer layers and dimensions of the hidden representations, respectively. $L=0$ corresponds to embedding-based query encoders, which are initialized with pre-trained token embeddings from \bertbase{}, and $L>0$ corresponds to attention-based query encoders, where the number of attention heads is $A=\frac{H}{64}$. The document encoder is a BERT model with \num{12} layers and \num{768}-dimensional representations in all cases. Query encoding latency is measured on \cpu{CPU} with a batch size of \num{256} queries from \msmpsgdev{} (tokenization cost is excluded, as it is identical for all models).}
  \label{fig:results.query_enc.psg_results}
\end{figure*}
In this section, we investigate the role of the query encoder in interpolation-based re-ranking using \fastforward{} indexes.

\subsubsection{The Role of Self-Attention}
First, we train a large number of dual-encoder models (as described in \cref{sec:setup.training}) and successively reduce the complexity of the query encoder. At the same time, we monitor the effects on performance and latency. The query encoders we analyze correspond to the \emph{attention-based query encoders} in \cref{sec:efficient_encoders.query.attn} and the \emph{embedding-based query encoders} in \cref{sec:efficient_encoders.query.embedding}. Since the embedding-based encoders are, technically speaking, a special case of the attention-based ones, we plot the results together in \cref{fig:results.query_enc.psg_results}. The document encoder we use is a \bertbase{} model, which has $L=12$ layers and $H=768$ hidden dimensions; it is the same across all experiments. For the query encoder, we start with \bertbase{} as well and reduce both the number of layers and hidden dimensions. All pre-trained BERT models we use for this experiment are provided by \citet{turc2019well}. If the output dimensions of the encoders do not match, we add a single linear layer to the query encoder (cf.\ \cref{sec:setup.training.architecture}).

\Cref{fig:results.query_enc.psg_results.latency_dev} illustrates the time each encoder requires to encode a batch of queries on a CPU; as expected, a reduction in either the number of layers or hidden dimensions has a positive impact on encoding latency, and the most lightweight attention-based model ($L=2$, $H=128$) is significantly faster than \bertbase{} (\num{27} milliseconds vs.\ \num{3.1} seconds). Furthermore, the complete omission of self-attention in the embedding-based encoder ($L=0$, $H=768$) results in even faster encoding (\num{13} milliseconds).

Next, we analyze to what extent the drastic reduction of complexity affects the ranking performance. \Cref{fig:results.query_enc.psg_results.ndcg_dev,fig:results.query_enc.psg_results.ndcg_psg19,fig:results.query_enc.psg_results.ndcg_psg20} show the corresponding \fastforward{} re-ranking performance on passage development and test sets. It is evident that the absolute difference in performance between the encoders is relatively low; this is especially true on \msmpsgdev{} and \trecdlpsgn{}. In fact, the embedding-based query encoder does not always yield worse performance than the attention-based encoders, specifically on \trecdlpsgn{}. On \trecdlpsgt{}, the highest absolute difference of \num{0.05} is the largest among the three datasets.

These results suggest that query encoders do not need to be overly complex; rather, in most cases, either considerably smaller attention-based or even embedding-based models can be used. The embedding-based encoders are particularly useful, since they are essentially a look-up table and hence require no forward pass other than computing the average of all token embeddings.

\subsection{What is the trade-off between \fastforward{} index size and ranking performance?}
\label{sec:results.index_size}
This research question investigates how index size influences ranking performance and latency. In detail, we reduce index size in two different ways: First, we apply sequential coalescing (cf.\ \cref{sec:ff_indexes.coalescing}) in order to reduce the \emph{number of vector representations} in the index. Second, we train query and encoders to output \emph{lower-dimensional vector representations}. Note that these methods are not mutually exclusive, but rather complementary.

\subsubsection{Sequential Coalescing}
In order to evaluate this approach, we first take the pre-trained \tct{} dense index of the MS MARCO corpus, apply sequential coalescing with varying values for $\delta$ and evaluate each resulting compressed index using the \trecdldocn{} test set.
\begin{figure}
  \centering
  \begin{tikzpicture}
    \pgfplotstableread[col sep=comma]{plots/coalescing/coalescing.csv}\data
    \begin{axis}[
            ymin=0,
            ymax=100,
            axis y line*=left,
            bar width=0.5cm,
            width=0.9\textwidth,
            height=0.3\textheight,
            every axis plot/.append style={fill},
            xlabel={Distance threshold $\delta$},
            xtick=data,
            xticklabels from table={\data}{1},
            ylabel={\%},
            grid=major,
        ]
        \addplot+[
            ybar,
            area legend,
            plotColorNeutral*,
            thick,
            draw=black,
            mark=none,
        ] table[x index=0, y index=2] {\data}; \label{plot1}
    \end{axis}

    \begin{axis}[
            ymin=0.3,
            ymax=0.8,
            axis y line*=right,
            width=0.9\textwidth,
            height=0.3\textheight,
            xtick=data,
            xticklabels from table={\data}{1},
            ylabel={Metric value},
            legend pos=north west,
            legend columns=-1,
            legend style={/tikz/every even column/.append style={column sep=0.5cm}},
        ]
        \addlegendimage{/pgfplots/refstyle=plot1}
        \addlegendentry{Size reduction [\%]}
        \addplot+[
            const plot mark mid,
            thick,
            plotColor1,
            mark=square,
        ] table[x index=0, y index=4] {\data};
        \addlegendentry{AP}
        \addplot+[
            const plot mark mid,
            thick,
            plotColor2,
            mark=o,
        ] table[x index=0, y index=5] {\data};
        \addlegendentry{nDCG}
    \end{axis}
\end{tikzpicture}
  \caption{Sequential coalescing applied to \trecdldocn{}. The plot shows the index size reduction in terms of the number of passages and the corresponding metric values for \fastforward{} interpolation with \tct{}.}
  \label{fig:results.index_size.coalescing}
\end{figure}
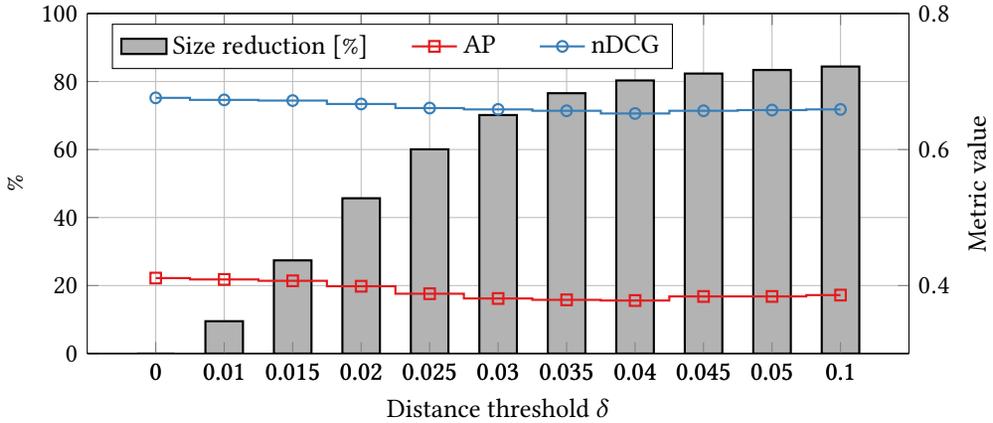
The results are illustrated in \cref{fig:results.index_size.coalescing}. It is evident that, by combining the passage representations, the number of vectors in the index can be reduced by more than \num{80}\% in the most extreme case, where only a single vector per document remains. At the same time, the performance is correlated with the granularity of the representations. However, the drops are relatively small. For example, for $\delta = 0.025$, the index size is reduced by more than half, while the nDCG decreases by roughly \num{0.015} (\num{3}\%).

Additionally, \cref{tab:results.rerank.model_eval_doc} shows the detailed performance of coalesced \fastforward{} indexes on the document datasets. We chose the indexes corresponding to $\delta = 0.035$ (\tct{}) and $\delta = 0.003$ (\ance{}), both of which are compressed to approximately 25\% of their original size. This is reflected in the query processing latency, which is reduced by more than half. The overall performance drops to some extent, as expected, however, these drops are not statistically significant in all but one case. The trade-off between latency (index size) and performance can be controlled by varying the threshold $\delta$.

\subsubsection{The Effect of Representation Size}
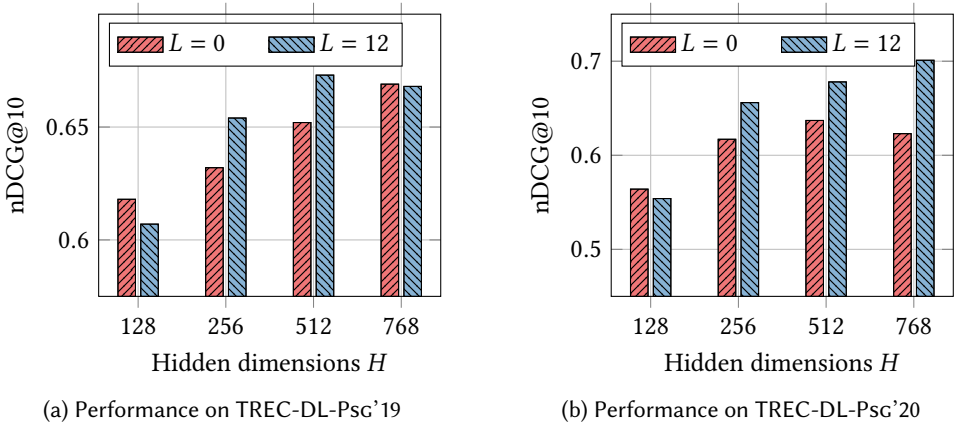
\begin{figure*}
  \begin{subfigure}{.49\linewidth}
    \centering
    \begin{tikzpicture}
    \begin{axis}[
            ybar,
            width=0.9\textwidth,
            bar width=0.2,
            every axis plot/.append style={fill},
            grid=major,
            xtick={1, 2, 3, 4},
            xticklabels={$128$, $256$, $512$, $768$},
            xlabel={Hidden dimensions $H$},
            ylabel={nDCG@\num{10}},
            enlarge x limits=0.15,
            ytick={0.6, 0.65},
            ymin=0.575,
            ymax=0.7,
            area legend,
            legend entries={$L = 0$, $L = 12$},
            legend cell align={left},
            legend pos=north west,
            legend columns=-1,
            legend style={/tikz/every even column/.append style={column sep=0.5cm}},
        ]
        \addplot+[
            ybar,
            plotColor1*,
            draw=black,
            postaction={
                    pattern=north east lines
                },
        ] plot coordinates {
                (1,0.618)
                (2,0.632)
                (3,0.652)
                (4,0.669)
            };
        \addplot+[
            ybar,
            plotColor2*,
            draw=black,
            postaction={
                    pattern=north west lines
                },
        ] plot coordinates {
                (1,0.607)
                (2,0.654)
                (3,0.673)
                (4,0.668)
            };
    \end{axis}
\end{tikzpicture}
    \caption{Performance on \trecdlpsgn{}}
    \label{fig:results.index_size.dim.psg19}
  \end{subfigure}
  \begin{subfigure}{.49\linewidth}
    \centering
    \begin{tikzpicture}
    \begin{axis}[
            ybar,
            width=0.9\textwidth,
            bar width=0.2,
            every axis plot/.append style={fill},
            grid=major,
            xtick={1, 2, 3, 4},
            xticklabels={$128$, $256$, $512$, $768$},
            xlabel={Hidden dimensions $H$},
            ylabel={nDCG@\num{10}},
            enlarge x limits=0.15,
            ymin=0.45,
            ymax=0.75,
            area legend,
            legend entries={$L = 0$, $L = 12$},
            legend cell align={left},
            legend pos=north west,
            legend columns=-1,
            legend style={/tikz/every even column/.append style={column sep=0.5cm}},
        ]
        \addplot+[
            ybar,
            plotColor1*,
            draw=black,
            postaction={
                    pattern=north east lines
                },
        ] plot coordinates {
                (1,0.564)
                (2,0.617)
                (3,0.637)
                (4,0.623)
            };
        \addplot+[
            ybar,
            plotColor2*,
            draw=black,
            postaction={
                    pattern=north west lines
                },
        ] plot coordinates {
                (1,0.554)
                (2,0.656)
                (3,0.678)
                (4,0.701)
            };
    \end{axis}
\end{tikzpicture}
    \caption{Performance on \trecdlpsgt{}}
    \label{fig:results.index_size.dim.psg20}
  \end{subfigure}
  \caption{\fastforward{} ranking results for $k_S=5000$ of embedding-based ($L=0$) and attention-based ($L=12$) query encoders. The representation dimension $H$ is always the same for both encoders. The document encoders use $L=12$ layers and $A=\frac{H}{64}$ attention heads in all cases.}
  \label{fig:results.index_size.dim}
\end{figure*}
In this experiment, we investigate the degree to which the dimension of the query and document representations influences the final ranking performance of the models. The idea is motivated by recent research~\cite{ni2021large}, which suggests that the representation vectors are not the bottleneck of dual-encoder models, but rather the document encoder complexity is. Since the dimensionality of the representations directly influences the index size, it is desirable to keep it as low as possible.

In order to analyze the effect, we train a number of dual-encoder models (cf.\ \ref{sec:prelims.dual_encoders.training}), where all hyperparameters except the hidden dimension $H$ and number of attention heads $A$ are kept the same. We show results for embedding-based ($L=0$) and attention-based ($L=12$) query encoders in \cref{fig:results.index_size.dim}. There is a trade-off between the dimensionality of representations and ranking performance, which is expected; this trade-off is exhibited by both embedding-based and attention-based query encoders. Overall, the results show that the performance reduction is rather small for $H=512$ and even $H=256$ (compared to $H=768$), considering that it goes hand in hand with a reduction in index size of approximately \num{33}\% and \num{67}\%, respectively.

\subsection{Can the indexing efficiency be improved by removing irrelevant document tokens?}
\label{sec:results.indexing_efficiency}
\begin{figure*}
  \begin{subfigure}{.32\linewidth}
    \centering
    \begin{tikzpicture}
    \pgfplotstableread[col sep=comma]{plots/selbert/selbert.csv}\data
    \begin{axis}[
            width=\textwidth,
            grid=major,
            xlabel={Ratio $p$},
        ]
        \addplot[
            domain=0.1:0.9,
            samples=2,
            dashed,
            mark=none,
            thick,
            plotColor1,
        ] {0.8979388236};
        \addplot+[
            const plot mark mid,
            thick,
            mark=square,
            plotColor2,
        ] table[y index=1] {\data};
    \end{axis}
\end{tikzpicture}
    \caption{Encoding latency [sec]}
    \label{fig:results.indexing_efficiency.sel_bert.latency}
  \end{subfigure}
  \begin{subfigure}{.32\linewidth}
    \centering
    \begin{tikzpicture}
    \pgfplotstableread[col sep=comma]{plots/selbert/selbert.csv}\data
    \begin{axis}[
            width=\textwidth,
            grid=major,
            ytick={0.25,0.3,0.35},
            xlabel={Ratio $p$},
            legend entries={\bertbase{},\selbert{}},
            legend cell align={left},
            legend pos=south east,
            legend style={font=\tiny},
        ]
        \addplot[
            domain=0.1:0.9,
            samples=2,
            dashed,
            mark=none,
            thick,
            plotColor1,
        ] {0.3661710229};
        \addplot+[
            const plot mark mid,
            thick,
            mark=square,
            plotColor2,
        ] table[y index=2] {\data};
    \end{axis}
\end{tikzpicture}
    \caption{nDCG@\num{10}, $L = 0$}
    \label{fig:results.indexing_efficiency.sel_bert.ndcg_emb}
  \end{subfigure}
  \begin{subfigure}{.32\linewidth}
    \centering
    \begin{tikzpicture}
    \pgfplotstableread[col sep=comma]{plots/selbert/selbert.csv}\data
    \begin{axis}[
            width=\textwidth,
            grid=major,
            ytick={0.25,0.3,0.35},
            xlabel={Ratio $p$},
        ]
        \addplot[
            domain=0.1:0.9,
            samples=2,
            dashed,
            mark=none,
            thick,
            plotColor1,
        ] {0.3885345674};
        \addplot+[
            const plot mark mid,
            thick,
            mark=square,
            plotColor2,
        ] table[y index=3] {\data};
    \end{axis}
\end{tikzpicture}
    \caption{nDCG@\num{10}, $L = 12$}
    \label{fig:results.indexing_efficiency.sel_bert.ndcg_attn}
  \end{subfigure}
  \caption{Evaluation of \fastforward{} indexes created using \selbert{} models. The document encoders are \bertbase{} models with $L=12$ and $H=768$. During fine-tuning, we set the parameter $p=0.75$ (percentage of tokens to keep). We then vary $p \in [0, 1]$ during the indexing stage, resulting in progressively higher indexing efficiency (\cref{fig:results.indexing_efficiency.sel_bert.latency}). The corresponding \fastforward{} ranking performance on \msmpsgdev{} is shown in \cref{fig:results.indexing_efficiency.sel_bert.ndcg_emb} for an embedding-based query encoder ($L=0$) and in \cref{fig:results.indexing_efficiency.sel_bert.ndcg_attn} for an attention-based query encoder ($L=12$). Document encoding latency is measured on \gpu{GPU} with a batch size of \num{256} passages from the MS MARCO corpus (tokenization cost is excluded, as it is identical for all models).}
  \label{fig:results.indexing_efficiency.sel_bert}
\end{figure*}
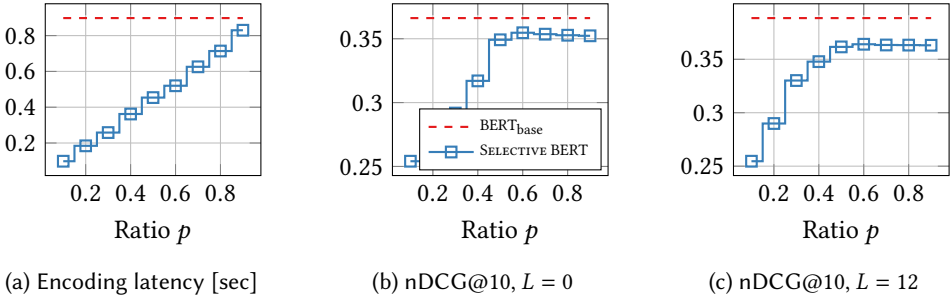
In this experiment, we focus on the \selbert{} document encoders proposed in \cref{sec:efficient_encoders.doc}. In order to analyze the index efficiency and ranking performance, we train two dual-encoders (cf.\ \cref{sec:setup.training}) with \selbert{} document encoders, where $L=12$ and $H=768$. The query encoders have $L=0$ (embedding-based) and $L=12$ (attention-based), respectively, and $H=768$. During fine-tuning (cf.\ \cref{sec:efficient_encoders.doc.ft_inf}), we fix the hyperparameter $p=0.75$, which controls the ratio of tokens to be removed from the documents; afterwards, we create a number of indexes, where we vary $p$ between \num{0.1} and \num{0.9}, and compute the corresponding indexing time (using GPUs) and ranking performance. The results are plotted in \cref{fig:results.indexing_efficiency.sel_bert}.

The document encoding latency (\cref{fig:results.indexing_efficiency.sel_bert.latency}) increases nearly linearly with the ratio of tokens to keep ($p$). Even though the BERT model has a quadratic complexity w.r.t.\ input length, this is expected, as there is a certain amount of overhead introduced by the scoring network and the reconstruction of the batches. More interestingly, the ranking performance (\cref{fig:results.indexing_efficiency.sel_bert.ndcg_emb,fig:results.indexing_efficiency.sel_bert.ndcg_attn}) is mostly unchanged for $p \geq 0.5$ in both cases, however, neither models manage to match the performance of their respective baselines (the same configuration with a standard BERT model instead of \selbert{}). We hypothesize that the reason for this could be the choice of $p=0.75$ during the fine-tuning step.

Overall, our results show that up to \num{50}\% of document tokens can be removed without much of a performance reduction. Encoding half of the number of tokens results in approximately halving the time required to encode documents. This has a large impact on efficient index maintenance in the context of dynamically increasing document collections. For future work, the \selbert{} architecture can be further refined, for example, by introducing improved (contextualized) scoring networks.

\section{Discussion}
\label{sec:discussion}
In this section, we reflect upon our work and present possible limitations.

\subsection{Efficient Encoders for Dense Retrieval}
\label{sec:discussion.retrieval}
\begin{table*}
    \centering
    \begin{tabular}{lcccccc}
        \toprule
                        & \multicolumn{3}{c}{\trecdlpsgn}
                        & \multicolumn{3}{c}{\trecdldocn}
        \\
        \cmidrule(lr){2-4}
        \cmidrule(lr){5-7}
                        & $\text{AP}_\text{1k}$           & $\text{R}_\text{1k}$ & $\text{nDCG}_\text{10}$
                        & $\text{AP}_\text{1k}$           & $\text{R}_\text{1k}$ & $\text{nDCG}_\text{10}$
        \\
        \midrule
        \multicolumn{7}{l}{\bfseries \sparseretrieval}                                                                                               \\
        \bm             & \num{0.301}                     & \num{0.750}          & \num{0.506}             & \num{0.331} & \num{0.697} & \num{0.519} \\
        \midrule
        \multicolumn{7}{l}{\bfseries \denseretrieval}                                                                                                \\
        \ance           & -                               & -                    & \num{0.648}             & -           & -           & \num{0.628} \\
        \tct            & -                               & -                    & \num{0.670}             & -           & -           & -           \\
        \midrule
        \multicolumn{7}{l}{\bfseries \textsc{Our Models}}                                                                                            \\
        $L=0$, $H=768$  & \num{0.198}                     & \num{0.486}          & \num{0.424}             & \num{0.100} & \num{0.263} & \num{0.342} \\
        $L=12$, $H=768$ & \num{0.318}                     & \num{0.691}          & \num{0.545}             & \num{0.201} & \num{0.457} & \num{0.504} \\
        \bottomrule
    \end{tabular}
    \caption{Retrieval results of dual-encoder models using lightweight query encoders and some baselines. For \trecdldocn{}, the dense retrieval depth is set to $k_D=10000$ and maxP aggregation is applied (cf.\ \cref{eqn:prelims.interpolated_reranking.maxp}). Our model with $L=0$ uses an embedding-based query-encoder, and the one with $L=12$ uses an attention-based query encoder. The document encoder is a $\text{BERT}_\text{base}$ model ($L=12$, $H=768$) in both cases.}
    \label{tab:discussion.retrieval.results}
\end{table*}
Our research questions and experiments have focused exclusively on interpolation-based re-ranking using dual-encoders and \fastforward{} indexes. However, the most common application of dual-encoders in the field of IR is the use as dense retrieval models; a natural question that occurs is, whether the encoders proposed in \cref{sec:efficient_encoders} can be used for more efficient dense retrieval.

In \cref{tab:discussion.retrieval.results}, we present passage and document retrieval results on the MS MARCO corpus. Dense retrievers use a \faiss{}~\cite{johnson2021billion} vector index; no interpolation or re-ranking is performed. It is immediately obvious that our models do not achieve competitive results; on the contrary, the embedding-based encoder yields far worse performance than dense retrievers and even \bm{}, and even the attention-based encoder fails to improve over sparse retrieval.

From these results, we infer that the models we trained are not suitable for dense retrieval. However, we assume that the main reason for this is not the architecture of the query encoder, but instead the following:
\begin{itemize}
    \item We use a simple in-batch negative sampling strategy~\cite{karpukhin2020dense}, which has been shown to be less effective than more involved strategies~\cite{zhan2021optimizing,xiong2021approximate,lin2020distilling,lindgren2021efficient}.
    \item The hardware we use for training the models is limiting w.r.t.\ the batch size and thus the number of negative samples, i.e., we cannot use a batch size greater than 4.
    \item We perform validation and early stopping based on re-ranking.
\end{itemize}
Considering the points above, we expect that our dual-encoder models, including ones with lightweight encoders, could also be used in retrieval settings if the shortcomings of the training setup are addressed, for example, by using more powerful hardware and state-of-the-art training approaches. On the other hand, we argue that the fact that our models perform well in the re-ranking setting (see \cref{sec:results}) shows that it is both easier and more efficient (in terms of time and resources) to train models to be used with \fastforward{} indexes instead of for dense retrieval.

\subsection{Out-of-Domain Performance}
\label{sec:discussion.ood}
\begin{table*}
    \centering
    \begin{tabular}{lccc}
        \toprule
                     &             & \multicolumn{2}{c}{\fastforward}                   \\
        \cmidrule(lr){3-4}
                     & \bm         & $L=0$, $H=768$                   & $L=12$, $H=768$ \\
        \midrule
        \beirmsm     & \num{0.477} & \num{0. 653}                     & \num{0.677}     \\
        \beirfever   & \num{0.649} & \num{0. 715}                     & \num{0.777}     \\
        \beirfiqa    & \num{0.254} & \num{0. 282}                     & \num{0.313}     \\
        \beirquora   & \num{0.808} & \num{0. 761}                     & \num{0.804}     \\
        \beirhpqa    & \num{0.602} & \num{0. 628}                     & \num{0.674}     \\
        \beirdbp     & \num{0.320} & \num{0. 331}                     & \num{0.393}     \\
        \beirscifact & \num{0.691} & \num{0. 676}                     & \num{0.698}     \\
        \beirnfc     & \num{0.327} & \num{0. 327}                     & \num{0.330}     \\
        \bottomrule
    \end{tabular}
    \caption{Zero-shot ranking results on BEIR datasets (nDCG@\num{10}) using embedding-based ($L=0$) and attention-based ($L=12$) query encoders. The document encoder is a BERT model with \num{12} layers and \num{768}-dimensional representations. The sparse retrieval depth is $k_S = 5000$.}
    \label{tab:discussion.ood.beir}
\end{table*}
In the previous sections, we found that \fastforward{} indexes and lightweight query encoders show good performance in in-domain ranking tasks. This raises the question whether the models generalize well to out-of-domain tasks.

In order to ascertain the out-of-domain capabilities of our models, we evaluate them on a number of test sets from the BEIR benchmark. The evaluation happens in a zero-shot fashion, meaning that we use the same models as before and do not re-train them on the respective datasets. The results are shown in \cref{tab:discussion.ood.beir}. It is apparent that the attention-based query encoder yields better results than the embedding-based one in all cases, but the difference varies across datasets. Since both models were trained on MS MARCO, they perform well on the BEIR version of that dataset, as expected; notable differences in performance are observed on \beirfever{} and \beirdbp{}, however, both models manage to improve the \bm{} results. Finally, on \beirquora{}, \beirscifact{} and \beirnfc{}, re-ranking does not lead to a performance improvement, but rather fails to improve or even degrades the results. We assume that the corresponding tasks either require specific in-domain knowledge of the model or would benefit greatly from query-document attention (cross-attention).

\subsection{Threats to Validity}
\label{sec:discussion.threats}
In this section, we outline and discuss certain aspects of the experimental evaluation in this article which result in possible threats to the validity of the results.

\subsubsection{Performance of \bertcls{}}
\label{sec:discussion.threats.bert}
In \cref{tab:results.rerank.model_eval_doc,tab:results.rerank.model_eval_passage}, we report the performance of dual-encoder ranking models, along with a cross-attention model (\bertcls{}). We found that \bertcls{} performed notably worse, especially when the sparse retrieval depth $k_S$ is increased. This result is unexpected, especially considering the fact that the cross-attention architecture allows for query-document attention.

In addition to the architecture itself, the models differ in the way they are trained: \ance{} and \tct{} use complex distillation and negative sampling approaches, along with contrastive loss functions (cf.~\cref{eq:prelims.dual_encoders.training.loss}), while \bertcls{} is trained using simple pairwise loss. It is thus reasonable to assume that the negative sampling approach has a positive impact on the performance. Specifically, the contrastive loss trains the models to identify relevant documents among a very large number of irrelevant documents, while the pairwise loss focuses on re-ranking mostly related documents, which could explain the performance drop for higher retrieval depths.

Furthermore, it is important to note that, even if \bertcls{} performed similarly to the dual-encoder models, the difference in efficiency would remain the same, leaving the claims we make unaffected.

\subsubsection{Latency Measurements}
\label{sec:discussion.threats.latency}
As \fastforward{} indexes aim at improving ranking efficiency, we mainly focus on the query processing latency, which is reported in \cref{tab:results.rerank.model_eval_doc,tab:results.rerank.model_eval_passage,tab:results.rerank.first_stage,fig:results.query_enc.psg_results}. As the experiments in the paper have been performed over a longer period of time, there have been slight changes with respect to, for example, hardware or implementations. Consequently, the numbers in latency might not be directly comparable \textbf{across experiments}. Thus, we made sure to make each experiment self-contained, such that these comparisons are not necessary; rather, our results highlight relative latency improvements \textbf{within} each experiment, where all measurements are comparable. In general, one should also keep in mind that latency can be heavily influenced by the way a method is implemented.

\subsubsection{Hybrid Retrieval Baselines}
\label{sec:discussion.threats.hybrid_retrieval}
In \cref{tab:results.rerank.model_eval_doc,tab:results.rerank.model_eval_passage}, we presented, along with the results of our own method, some hybrid retrieval baselines. \Cref{tab:setup.evaluation.encoders_indexes} shows the corresponding indexes that we used for the dense retrievers. It is important to note that those are \emph{brute-force} indexes, i.e., they perform exact $k$NN retrieval. It is thus to be expected that the latency of hybrid retrieval can be further reduced by employing approximate dense retrieval instead; this would likely go hand in hand with a reduction in performance though.

\section{Conclusion}
\label{sec:conclusion}
In this paper, we proposed \fastforward{} indexes, a simple yet effective and efficient look-up-based interpolation method that combines lexical and semantic ranking. \fastforward{} indexes are based on dense dual-encoder models, exploiting the fact that document representations can be pre-processed and stored, providing efficient access in constant time. Using interpolation, we observed increased performance compared to hybrid retrieval. Furthermore, we achieved improvements of up to \num{75}\% in memory footprint and query processing latency due to our optimization techniques, \emph{sequential coalescing} and \emph{early stopping}.

Moreover, we introduced efficient encoders for dual-encoder models: Embedding-based and lightweight attention-based query encoders can be used to compute query representations significantly faster without compromising performance too much. \selbert{} document encoders dynamically remove irrelevant tokens from input documents prior to indexing, reducing the document encoding latency by up to \num{50}\% and thus making index maintenance much faster.

Our method solely requires CPU computations for ranking, completely eliminating the need for expensive GPU-accelerated re-ranking.

\begin{acks}
  This work is supported by the European Union – Horizon 2020 Program under the scheme “INFRAIA-01-2018-2019 – Integrating Activities for Advanced Communities”, Grant Agreement n.871042, “SoBigData++: European Integrated Infrastructure for Social Mining and Big Data Analytics” (http://www.sobigdata.eu).

  This work is supported in part by the Science and Engineering Research Board, Department of Science and Technology, Government of India, under Project SRG/2022/001548 and Microsoft Academic Partnership Grant 2023 Agreement No.\ 7581365. Koustav Rudra is a recipient of the DST-INSPIRE Faculty Fellowship [DST/INSPIRE/04/2021/003055] in the year 2021 under Engineering Sciences.
\end{acks}

\bibliographystyle{ACM-Reference-Format}
\bibliography{references}

\end{document}